\documentclass[a4paper,12pt,reqno]{amsart}%{article}
\usepackage{latexsym,amsmath,amssymb,amsthm}
\usepackage{enumerate}
\usepackage{verbatim}
\usepackage[dvipsnames]{xcolor}
\usepackage{bbm}
\usepackage{mathrsfs}
\usepackage[a4paper, textwidth=16cm, textheight=23.1cm]{geometry}
\usepackage[colorlinks=true, linkcolor=BrickRed, urlcolor=Blue,citecolor=OliveGreen,menucolor=]{hyperref}

\newtheorem{theorem}{Theorem}
\newtheorem{proposition}[theorem]{Proposition}
\newtheorem{corollary}[theorem]{Corollary}
\newtheorem{lemma}[theorem]{Lemma}

\theoremstyle{definition}
\newtheorem{definition}[theorem]{Definition}

\theoremstyle{remark}

\newtheorem{remark}[theorem]{Remark}
\newtheorem{example}[theorem]{Example}

\numberwithin{theorem}{section}
\numberwithin{equation}{section}

\newcommand{\1}{{\bf 1}}

\newcommand{\id}{{\rm id}}

\newcommand{\pb}{\{\,\cdot\,,\,\cdot\,\}}
\newcommand{\Id}{\mathbbm 1}

\DeclareMathOperator{\ad}{\textnormal{ad}}
\DeclareMathOperator{\Ad}{\textnormal{Ad}}

\DeclareMathOperator{\Ker}{\textnormal{Ker}}

\DeclareMathOperator{\Tr}{\textnormal{Tr}}

\DeclareMathOperator{\im}{\textnormal{Im}}

\newcommand{\Nc}{\mathcal{N}}
\newcommand{\Cp}{\mathscr{C}}

\newcommand{\g}{{\mathfrak g}}
\newcommand{\bg}{{\mathfrak b}}

\newcommand\R{{\mathbb{R}}}
\newcommand\C{{\mathbb{C}}}

\newcommand\N{{\mathbb{N}}}
\newcommand\K{{\mathbb{K}}}

\renewcommand\a{{\mathbf{a}}}
\renewcommand\b{{\mathbf{b}}}
\newcommand\p{{\mathbf{p}}}
\newcommand\q{{\mathbf{q}}}

\begin{document}
\bibliographystyle{grz2}

\title[Banach Poisson--Lie groups, Lax equations and AKS theorem]{Banach Poisson--Lie groups, Lax equations and the AKS theorem in infinite dimensions}

\author[T.~Goli\'nski]{Tomasz Goli\'nski}
\address{University of Bia\l ystok\\ Cio\l kowskiego 1M\\15-245 Bia\l ystok\\ Poland}
\email{tomaszg@math.uwb.edu.pl}
\author[A.B.~Tumpach]{Alice Barbora Tumpach}
\address{Institut CNRS Pauli\\ UMI  CNRS 2842\\ Oskar-Morgenstern-Platz,1 \\1090 Wien\\Austria\\ \& UMR CNRS 8524\\
UFR de Math\'ematiques\\
Laboratoire Paul Painlev\'e\\
59 655 Villeneuve d'Ascq Cedex\\
France }
\email{alice-barbora.tumpach@univ-lille.fr}

\begin{abstract}
In this paper, we investigate the theory of $R$-brackets, Baxter brackets and Nijenhuis brackets in the Banach setting, in particular in relation with Banach Poisson--Lie groups. The notion of Banach Lie--Poisson space with respect to an arbitrary duality pairing  is crucial for the equations of motion to make sense. In the presence of a non-degenerate invariant pairing on a Banach Lie algebra, these equations of motion assume a Lax form. We prove a version of the Adler--Kostant--Symes theorem adapted to $R$-matrices on infinite-dimensional Banach algebras. Applications to the resolution of Lax equations associated to some Banach Manin triples are given. The semi-infinite Toda lattice is also presented as an example of this approach. 


\end{abstract}

\date{\today}

\maketitle

\tableofcontents

\section{Introduction}

\subsection{Motivation}

The Adler--Kostant--Symes (AKS) theorem is a fundamental theorem in the theory of Hamiltonian systems. It allows to associate to a splitting of a finite-dimensional Lie algebra a new Lie bracket leading to isospectral evolutions,  i.e. with spectral functions in involutions. Although the AKS theorem and similar involutivity theorems have been extensively applied both in the finite and  infinite-dimensional setting (see \cite{adler78,symes80int,symes80,ratiu80involution,morrison2013,pelletier24} and the references therein), a non-formal presentation in the Banach setting seems to be lacking. Recent applications to the theory of Lie group thermodynamic \cite{Barbaresco_AKS} motivate us to think about the foundations of this theory in the Banach case.

In the present paper we develop the  Banach version of notions related to the theory of $R$-matrices, Rota-Baxter algebras and Nijenhuis operators, in particular in relation with Banach Poisson--Lie groups \cite{tumpach-bruhat,GT-u-bial}. The notion of Banach Lie--Poisson space with respect to an arbitrary duality pairing  is crucial for the equations of motion to make sense. In the presence of an invariant non-degenerate pairing on a Banach Lie algebra, these equations of motion can be written as Lax equations.  We prove a version of the Adler--Kostant--Symes theorem adapted to $R$-matrices on infinite-dimensional Banach algebras (Theorem~\ref{integrale_curve}). This theorem is then applied to Manin triples of Banach Lie algebras in Schatten classes related to Iwasawa decompositions of the corresponding groups. 
The semi-infinite Toda lattice is also presented in link with this Banach theory.

\subsection{Structure of the paper}

The first section contains a summary of the theory of Banach Poisson--Lie groups developed in \cite{tumpach-bruhat}. It is as self-contained as possible and can be used as a first introduction to the subject. In particular, the equivalence between Manin triples and Banach Lie--Poisson spaces which are Banach Lie bialgebras is presented. This equivalence is at the heart of the necessity to extend the notion of Banach Poisson manifolds to the one presented in section~\ref{sec:BPLG}. In sections~\ref{sec:R}, \ref{sec:B} and \ref{sec:N} we recall different approaches that lead to the definition of an auxiliary Lie bracket on a Banach Lie algebra using an operator satisfying some equations, like the modified classical Yang--Baxter equation, the Baxter equation or the vanishing of the Nijenhuis torsion. This auxiliary bracket can lead to the existence of a new structure of Banach Lie--Poisson space on any space in duality with the original Banach Lie algebra. Section~\ref{sec:Involutivity} contains the involutivity theorems that we prove in the Banach context. In section~\ref{sec:Lax}, the equations of motion on coadjoint orbits are transported to adjoint orbits using an $\Ad_G$-invariant pairing, leading to equations in Lax form. The resolution of these equations using the solution of the factorization problem is presented and applied to the Iwasawa decomposition. In section~\ref{sec:Toda}, we present how the theory allows to recover the equations of the semi-infinite Toda lattice in Flaschka coordinates.

\subsection{Notation and basic properties}
Let $\mathcal H$ be a complex separable Hilbert space. For a bounded linear operator $A\in L^\infty(\mathcal H)$, the square root of $A^*A$ is well defined, and denoted by $(A^*A)^{\frac{1}{2}}$ (see  \cite[Theorem VI.9]{reed1}).
The Schatten class $L^{p}(\mathcal{H})$ is the subspace of bounded operators $A$ such that
\[\|A\|_p = \left(\Tr (A^*A)^{\frac{p}{2}}\right)^\frac{1}{p}\]
is finite. For $p\geq 1$, it is a Banach Lie algebra with the norm $\|\cdot \|_p$ and the bracket given by the commutator of operators. In particular, $L^1(\mathcal{H})$ will denote the Banach Lie algebra of trace class operators, and $L^2(\mathcal{H})$ will denote the Hilbert Lie algebra of Hilbert--Schmidt operators. We recall that $L^{p}(\mathcal{H})$ is a two-sided ideal in $L^\infty(\mathcal H)$, i.e. for any $A \in L^{p}(\mathcal{H})$ and $B \in L^{\infty}(\mathcal{H})$, $AB, BA \in L^{p}(\mathcal{H})$.

Moreover, $L^p(\mathcal H)$ is a Banach Lie algebra of the Banach Lie group
\[GL^p(\mathcal H) = (\Id + L^p(\mathcal H))\cap GL(\mathcal H),\]
where $\Id$ denotes the identity operator on $\mathcal H$.
For the remainder of the paper we fix $p$ and $q$ such that $1< p\leq q< \infty$ and $\frac{1}{p}+\frac{1}{q} = 1$.
Recall that for $x\in L^p(\mathcal{H})$ and $\alpha\in L^q(\mathcal{H})$, the operator $x\alpha$ is trace class and  
\[
\|x\alpha\|_1\leq \|x\|_p\|\alpha\|_q,
\]
(see Proposition~5, page~41 in \cite{reed2}). 
Moreover  $L^p(\mathcal{H})^* = L^{q}(\mathcal{H})$ 
by the strong duality pairing given by the trace
\[
\begin{array}{lcll}
\Tr: & L^p(\mathcal{H})\times L^{q}(\mathcal{H})& \longrightarrow & \mathbb{C}\\
& (x, \alpha) & \longmapsto & \Tr\left(x\alpha\right),
\end{array}
\]
(see Proposition~7, page~43 in \cite{reed2} and Theorem~VI.26, page~212 in \cite{reed1}). Using the invariance of the trace under cyclic permutations $\Tr(AB) = \Tr(BA)$ for $A\in L^1(\mathcal{H})$ and $B\in L^\infty(\mathcal{H})$ (see Theorem~VI.25, page~212 in \cite{reed1}), for any $\alpha, \beta\in L^{q}(\mathcal{H})$ and any $x\in L^p(\mathcal{H})$, one has
\begin{equation}\label{ad*}
\ad^*_\alpha x(\beta) = \Tr x [\alpha, \beta] = -\Tr\left([\alpha, x]\beta\right),
\end{equation}
where the bracket 
is the commutator of the bounded linear operators $x\in L^{p}(\mathcal{H})$ and $\alpha\in L^{q}(\mathcal{H})$.

\section{Banach Poisson--Lie groups in a nutshell}

We begin by recalling some basic definitions from Banach Poisson geometry originating from \cite{OR}, which were developed further and applied e.g. in \cite{Oext,beltita05,Ratiu-grass,Oind,GO-grass,GT-momentum,pelletier}. We also present a condensed version of the theory of Banach Poisson--Lie groups developed in \cite{tumpach-bruhat} and used in \cite{GT-u-bial}. For the comparison of different definitions of Poisson structures in the infinite-dimensional setting, we refer the reader to ~\cite{GRT-poisson-bial}.

\subsection{Notions of Poisson manifolds in the Banach setting}
The usual definition of a Poisson structure is the following. We will extend this definition to a subalgebra of admissible functions in section~\ref{sec:generalisation_poisson}.
\begin{definition}\label{condition_poisson}
On the space $\Cp^{\infty}(M)$ of smooth real-valued functions on a Banach manifold $M$, a $\mathbb{R}$-bilinear operation $\{\cdot,\cdot\}: \Cp^{\infty}(M)\times\Cp^{\infty}(M)\rightarrow\Cp^{\infty}(M)$ is called a \textbf{Poisson bracket} on $M$ if it satisfies:
\begin{enumerate}
\item[(i)] anti-symmetry: $\{F, G\} = -\{G, F\};$
\item[(ii)] Jacobi identity: $\{\{F, G\}, H\} + \{\{G, H\}, F\} + \{\{H, F\}, G\} =  0$;
\item[(iii)] Leibniz formula: $\{F, GH\} = \{F, G\}H + G\{F, H\}$.
\end{enumerate}
\end{definition}

We will use the notion of tensor and wedge products of Banach spaces as multilinear maps. In particular, for any Banach manifold $M$ the vector bundle $\Lambda^2T^{**}M$ is defined as the fiber bundle of \textbf{skew-symmetric continuous bilinear maps} on the cotangent bundle $T^*M$.

\begin{definition}
Given a Poisson structure $\{\cdot, \cdot\}$ on a Banach manifold $M$, a smooth section $\pi$ of the vector bundle $\Lambda^2T^{**}M$  satisfying 
\[\{F, G\} = \pi(DF, DH),\]
where $DF$ and $DG$ denote the Fr\'echet derivative of the smooth maps $F, G\in\Cp^{\infty}(M)$, is called a \textbf{Poisson tensor} associated to the Poisson structure $\{\cdot,\cdot\}$.

The vector bundle map $\sharp:T^*M \longrightarrow T^{**}M$ covering the identity defined by 
\[\sharp_m(\alpha_m) := \pi_m(\cdot, \alpha_m)\]
is called \textbf{Poisson anchor}.
\end{definition}

\begin{remark}
It is noteworthy to mention that in the infinite-dimensional case, a Poisson tensor might not exist for a Poisson bracket. An example of ``queer'' Poisson bracket depending on higher order differential on a Hilbert space (thus not given by a Poisson tensor) was constructed in \cite{BGT}. It is based on the existence of derivations of order greater than one (i.e. depending on higher order differential of functions than the first derivative), called ``queer'' vectors in \cite{michor}. Poisson brackets constructed using higher order derivations were therefore called queer. The existence of such Poisson tensors contradicts the belief that the Leibniz rule implies the existence of a Poisson tensor. 
\end{remark}

\begin{remark}
To the best of our knowledge it is not even known if
Poisson brackets need to be localizable, i.e. depend only on the germ of functions at a particular point, see \cite{BGT}. In the finite dimensional case this fact follows from Leibniz property and the existence of bump functions (i.e. non-zero functions with compact support). However on Banach manifolds (or even on Banach spaces) there may be no bump functions, see \cite{kurzweil1954,bonic-frampton1966} or discussion in \cite{pelletier}.
\end{remark}

Let us recall the following definition of Banach Poisson manifold given in \cite{OR} (with further clarifications from \cite{BGT}). In some cases more generalized definitions are needed, for example for the study of the restricted Grassmannian and the KdV equation, see e.g. \cite{tumpach-bruhat,pelletier,neeb14}. A discussion and comparison of various possible approaches can be read in \cite{GRT-poisson-bial}. We therefore start with the classical (restrictive) definition of Banach Poisson manifolds, and we will then drop some of the assumptions in order to be able to study more complex examples.
\begin{definition}[\cite{OR}, \cite{BGT}]\label{def:poissonmanifold}
A \textbf{Banach Poisson manifold} is a pair $(M, \pb)$ consisting of a smooth Banach manifold $M$ and a Poisson bracket $\pb$ given by a Poisson tensor $\pi$, such that 
the Poisson anchor $\sharp:T^*M \longrightarrow T^{**}M$ satisfies the condition 
\begin{equation}\label{BanachPoissonCompatibility}
\sharp(T^{*}M) \subset TM,
\end{equation}
where $TM$ is considered as a subbundle of $T^{**}M$ via the canonical injections of the fibers $T_mM\subset T_m^{**}M$.
\end{definition}

\begin{remark}
The compatibility condition \eqref{BanachPoissonCompatibility} is satisfied automatically if the modeling Banach space is reflexive. It allows to define, for any smooth function $H\in\Cp^{\infty}(M)$, the associated \textbf{Hamiltonian vector field} $X_{H}:= \sharp(DH) \in \Gamma(TM)$ which acts on $\Cp^{\infty}(M)$ by the following derivation
\[
X_H(F) = \langle DF, X_H \rangle = \{ F, H\} \quad\quad \forall F\in\Cp^{\infty}(M),
\]
where $ \langle \cdot,\cdot \rangle$ denotes the duality pairing  between fibers of $T^*M$ and $TM$.
\end{remark}

\subsection{Banach Lie--Poisson spaces}
A fundamental class of Banach Poisson manifolds needed in the present paper are the Banach Lie--Poisson spaces, which were introduced in the paper \cite{OR}, see Definition 4.1 and Theorem 4.2 therein. The notion was also extended to arbitrary duality pairing in \cite{tumpach-bruhat}, see  Definition~\ref{def:Banach_Lie_poisson_general} below. 
Recall that a Banach Lie algebra $\mathfrak{g}$ acts on itself and on its continuous dual $\mathfrak{g}^*$ by the adjoint and coadjoint actions:
\[
\begin{array}{l}
\begin{array}{llll}
\ad: &\mathfrak{g}\times\mathfrak{g}&\longrightarrow &\mathfrak{g}\\
& (x, y) & \longmapsto & \ad_x y := [x , y ],
\end{array}
\\
\\
\begin{array}{llll}
\ad^*: &\mathfrak{g}\times\mathfrak{g}^*&\longrightarrow &\mathfrak{g}^*\\
& (x, \alpha) & \longmapsto & \ad^*_x \alpha := \alpha\circ \ad_x.
\end{array}
\end{array}
\]

\begin{definition}\label{def:blp}
A \textbf{Banach Lie--Poisson} space is a Banach space $\g_*$ predual to a Banach Lie algebra $\g$ such that $\g_*\subset \g^*$ is preserved by the coadjoint action of $\g$
\begin{equation}\label{blp-condition}
\ad^*_\g \g_*\subset \g_*,
\end{equation}
together with the canonical structure of Banach Poisson manifold given by the bracket
\[ \label{blp-pb}
\{F,G\}(\mu) = \langle \mu, [D_\mu F,D_\mu G]_\g \rangle \]
for $F,G\in \Cp^\infty(\g_*)$. 
\end{definition}

In the formula above we treat the derivatives $D_\mu F$ and $D_\mu G$ at point $\mu$ as elements of the Banach Lie algebra $(\g_*)^*=\g$. The Hamiltonian vector field for a Hamiltonian $H\in \Cp^\infty(\g_*)$ with respect to this bracket assumes the form
\begin{equation} \label{blp-hvf}
X_H(\mu) = -\ad^*_{D_\mu H}\mu.
\end{equation}

\begin{example}\label{exLp}
    Since $L^p(\mathcal{H})$ is a reflexive Banach space, it is automatically a Banach Lie--Poisson space. A less trivial example is the space of trace-class operators $L^1(\mathcal{H})$, which is a predual space of all bounded operators $L^\infty(\mathcal{H})$.
\end{example}

\begin{remark}\label{existence_predual}
    In general, a closed subspace of a Banach space admitting a predual might not admit a predual. For instance the subspace of compact operators on a Hilbert space is a closed subspace of the Banach space of bounded operators which does not admit a predual, whereas the Banach space of bounded operator has the space of trace class operators as predual. Even if the predual does exist, it might not be unique and it is not guaranteed that it will be preserved by coadjoint action. Thus if $\g_*$ is a Banach Lie--Poisson space predual to $\g$, and $\g_+\subset \g$ is a closed Lie subalgebra, there might not be a Banach Lie--Poisson space predual to $\g_+$. See also the discussion in the context of precotangent bundles in \cite{G-precot}.
\end{remark}

\subsection{Generalized Banach Poisson manifolds}\label{sec:generalisation_poisson}
\begin{definition}\label{subbundleduality} 
We will say that $\mathbb{F}$ is a subbundle  of $T^*M$ \textbf{in duality} with the tangent bundle $TM$ of a Banach manifold  $M$ if, for every $x\in M$, 
\begin{enumerate}
\item $\mathbb{F}_x$ is an injected Banach space of $T_x^*M$, i.e. $\mathbb{F}_x$ admits a Banach space structure such that the injection $\mathbb{F}_x\hookrightarrow T_x^*M$ is continuous,
\item the natural duality pairing between $T_x^*M$ and $T_xM$ restricts to a duality pairing between $\mathbb{F}_x$ and $T_xM$, i.e.  $\mathbb{F}_x$ separates points in $T_xM$.
\end{enumerate}

\end{definition}

We will denote by $\Lambda^2\mathbb{F}^{*}$ the vector bundle over $M$ whose fiber over $x\in M$ is the Banach space of continuous skew-symmetric bilinear maps on the subspace $\mathbb{F}_x$ of $T_x^*M$.

\begin{definition}\label{Poisson_tensor}
Let $M$ be a Banach manifold and $\mathbb{F}$ a subbundle of $T^*M$ in duality with $TM$. A smooth section $\pi$ of $\Lambda^2\mathbb{F}^*$  is called a \textbf{Poisson tensor} on $M$ with respect to $\mathbb{F}$ if:
\begin{enumerate}
\item for any closed local sections $\alpha$, $\beta$ of  $\mathbb{F}$, the differential $D\left(\pi(\alpha, \beta)\right)$ is a local section of $\mathbb{F}$;
\item (Jacobi) for any closed local sections $\alpha$, $\beta$, $\gamma$ of $\mathbb{F}$,
\begin{equation}\label{Jacobi_Poisson}
\pi\left(\alpha, D\left(\pi(\beta, \gamma)\right)\right) + \pi\left(\beta, D\left(\pi(\gamma, \alpha)\right)\right) + \pi\left(\gamma, D\left(\pi(\alpha, \beta)\right) \right)= 0.
\end{equation}
\end{enumerate}
\end{definition}

\begin{definition}\label{Poisson-Manifold}
A \textbf{generalized Banach Poisson manifold} is a triple $(M, \mathbb{F}, \pi)$ consisting of a smooth Banach manifold $M$, a subbundle $\mathbb{F}$ of the cotangent bundle $T^*M$ in duality with $TM$, and a Poisson tensor $\pi$ on $M$ with respect to $\mathbb{F}$. On the unital subalgebra $\mathcal{A}\subset\Cp^{\infty}(M)$ consisting of smooth functions on $M$ with differentials in $\mathbb F$
\begin{equation}\label{A}
\mathcal{A} :=\{ F \in \Cp^\infty(M): D_xF \in \mathbb F_x \textrm{ for any }x\in M\},
\end{equation}
one can define the bracket of two functions $F, G\in\mathcal{A}$ by
\begin{equation}\label{pipoisson}
\{F, G\}(x) := \pi_x(D_xF, D_xG).
\end{equation}
Then $\{\cdot,\cdot\}:\mathcal{A}\times\mathcal{A}\rightarrow\mathcal{A}$ satisfies conditions $(i)-(iii)$ from Definition~\ref{condition_poisson} and is called a generalized Poisson bracket on $M$. 
\end{definition}

\subsection{Manin triples of Banach Lie algebras}

In the finite-dimensional theory, Manin triples are in one to one correspondence with Lie bialgebras and with connected and simply connected Poisson--Lie groups. Let us recall the notion of Manin triples in the Banach setting and review their link to Banach Lie bialgebras and Banach Poisson--Lie groups. See \cite{tumpach-bruhat} for more details.

\begin{definition}
A Banach Manin triple consists of a triple of Banach Lie algebras $(\mathfrak{g}, \mathfrak{g}_+, \mathfrak{g}_-)$ over a field $\mathbb{K}$ and a \textbf{non-degenerate symmetric bilinear} continuous map $\langle\cdot, \cdot\rangle_{\!\mathfrak{g}}:\g\times\g\rightarrow \mathbb{K}$ on $\mathfrak{g}$ such that 
\begin{enumerate}
\item the bilinear map $\langle\cdot,\cdot\rangle_{\!\mathfrak{g}}$ is invariant with respect to the bracket $[\cdot, \cdot]_{\mathfrak{g}}$ of $\mathfrak{g}$, i.e.
\begin{equation}\label{invariance_bilinear_map}
\langle [x, y]_{\mathfrak{g}}, z\rangle_{\!\mathfrak{g}}  + \langle y, [x, z]_{\mathfrak{g}}\rangle_{\!\mathfrak{g}} = 0,~~\forall x, y, z \in \mathfrak{g};
\end{equation}
\item $\mathfrak{g} = \mathfrak{g}_+\oplus\mathfrak{g}_-$ as Banach spaces;
\item both $\mathfrak{g}_+$ and $\mathfrak{g}_-$ are Banach Lie subalgebras of $\mathfrak{g}$;
\item both $\mathfrak{g}_+$ and $\mathfrak{g}_-$ are isotropic with respect to the bilinear map $\langle\cdot,\cdot\rangle_{\!\mathfrak{g}}$.
\end{enumerate}
\end{definition}

\begin{example}[Manin triples related to Iwasawa decompositions]\label{exIwasawa}
We will use the following notation.
The real Banach Lie algebra $\mathfrak{u}_{p}(\mathcal{H})$ is the Lie algebra of skew-Hermitian operators in $L^p(\mathcal{H})$:
\begin{equation}\label{u2}
\mathfrak{u}_p(\mathcal{H}) := \{A\in L^{p}(\mathcal H): A^* = -A\}.\quad\quad\quad\quad\quad\quad\quad\quad\quad\quad\quad\quad\quad\quad\quad\quad\quad\quad\quad
\end{equation} 
The real Banach subalgebra $\mathfrak{b}_p(\mathcal{H})$ 
is the triangular Banach algebra defined as follows: 

\begin{equation}\label{b2pm}
\begin{array}{l}
\mathfrak{b}_p(\mathcal{H}) := \{\alpha\in L^p(\mathcal{H}): \alpha |n\rangle \in~ \textrm{span}\{|m\rangle, m\geq n\}~\textrm{and}~\langle n|\alpha|n\rangle\in\mathbb{R}, \textrm{for}~ n\in\mathbb{Z}\},
\end{array}
\end{equation}
where $\{|n\rangle, n\in\mathbb{Z}\}$ is a fixed basis of $\mathcal{H}$.
\begin{proposition}[{\cite[Proposition~1.16]{tumpach-bruhat}}]\label{triples}
For $1<p\leq 2$, the triples of Banach Lie algebras $(L^p(\mathcal{H}), \mathfrak{u}_{p}(\mathcal{H}), \mathfrak{b}_p(\mathcal{H}))$ 
are real Banach Manin triples with respect to the pairing given by the imaginary part of the trace
\begin{equation}\label{imparttrace}
\begin{array}{lcll}
\langle\cdot,\cdot\rangle_{\mathbb{R}}: &L^p(\mathcal{H})\times L^p(\mathcal{H})& \longrightarrow &\mathbb{R}\\
& (x, y) & \longmapsto & \im \Tr\left(x y\right).
\end{array}
\end{equation}
\end{proposition}
\end{example}

\subsection{Banach Lie--Poisson spaces for an arbitrary duality pairing}

In order to relate Banach Manin triples with Banach Poisson--Lie groups and their infinitesimal versions, we will need a generalization of the notion of Banach Lie--Poisson space for an arbitrary duality pairing between two Banach spaces. Recall that a duality pairing $\langle\cdot, \cdot\rangle_{\mathfrak{b}, \mathfrak{g}}:\mathfrak{b}\times\mathfrak{g}\rightarrow\mathbb{K}$ between two Banach spaces over a field $\K$  is a non-degenerate continuous bilinear map. Note that a duality pairing between $\b$ and $\g$ allows to inject continuously $\b$ into the dual of $\g$, and $\g$ into the dual of $\b$.
\begin{definition}\label{def:Banach_Lie_poisson_general}
Consider a duality pairing $\langle\cdot, \cdot\rangle_{\bg, \g}:\bg\times\g\rightarrow\mathbb{K}$ between two Banach spaces.
We will say that $\bg$ is a \textbf{Banach Lie--Poisson space with respect to} $\g$ if $\g$ is a Banach Lie algebra $(\g, [\cdot,\cdot])$, which acts continuously on $\bg\hookrightarrow \g^*$ by coadjoint action, i.e.
\[
\ad^*_\alpha x \in\bg
\]
for all $x \in \bg$ and $\alpha\in\g$, and $\ad^*:\g\times\bg\rightarrow \bg$ is continuous.
\end{definition}

\begin{theorem}[{\cite[Theorem~3.14]{tumpach-bruhat}}]\label{4.2}
Suppose that $\bg$ is a Banach Lie--Poisson space with respect to $\g$. Denote by $\mathbb{F}$ the subbundle of $T^*\bg\simeq \bg\times\bg^*$ with the fiber at $x\in\bg$ given by
\[
\mathbb{F}_x = \{x\}\times\g\subset \{x\}\times\bg^*\simeq T_x^*\bg.
\] 
For any two local closed sections $\alpha$ and $\beta$ of $\mathbb{F}$, define a tensor $\pi\in\Lambda^2\mathbb{F}^*$ by:
\[
\pi_x(\alpha, \beta) := \left\langle x, [\alpha(x), \beta(x)]\right\rangle_{\bg,\g}.
\]
Then $(\bg, \mathbb{F}, \pi)$ is a generalized Banach Poisson manifold and $\pi$ takes values in $\Lambda^2\bg\subset\Lambda^2\mathbb{F}^*$. The unital subalgebra $\mathcal{A}\subset\Cp^{\infty}(\bg)$ defined by \eqref{A} consists of all functions with differentials in $\g$:
\begin{equation}\label{Ag}
 \mathcal{A} =\{ F \in \Cp^\infty(\bg): D_xF \in \g\subset \bg^* \textrm{ for any }x\in\bg\}.
\end{equation}

The generalized Poisson bracket of two functions $F, G\in\mathcal{A}$ takes the form
\begin{equation}\label{pipoisson-blp}
\{F, G\}(x) := \pi_x(D_xF, D_xG) = \left\langle x, [D_xF, D_xG]\right\rangle_{\bg, \g}.
\end{equation}

The Hamiltonian vector field  associated with $H\in \mathcal{A}$ is given by
\[
X_H(x) = -\ad^*_{D_xH}x \in\bg.
\]
\end{theorem}

A particular case of previous theorem arises when a Banach Lie algebra $\g$ of a Banach Lie group $G$ admits an invariant non-degenerate continuous bilinear map $\langle\cdot,\cdot\rangle:\g\times\g\rightarrow\g$, in the sense that 
\[
\langle [x_1, x_2], x_3\rangle + \langle x_2, [x_1, x_3]\rangle = 0 \quad\quad\forall x_1, x_2, x_3 \in \mathfrak{g}.
\]
In this case we have the following.

\begin{corollary}\label{g_lie_poisson}
    Suppose that a Banach Lie algebra $\g$ of a Banach Lie group $G$ admits a non-degenerate continuous bilinear map $\langle\cdot,\cdot\rangle:\g\times\g\rightarrow\g$, invariant by the adjoint action of $\g$, and denote by $\iota:\g\rightarrow \g^*$ the injection which maps $X\in\g$ to $\langle X, \cdot\rangle\in\g^*$. Then $\g$ is a Banach Lie--Poisson space with respect to itself.
 
The Hamiltonian vector field associated to a smooth function $H$ in $\mathcal A$

is given by
\[
X_H(x) = [D_xH, x] \in\mathfrak{g}.
\]
\end{corollary}

\begin{proof}
The fact that $\g$ is a Banach Lie--Poisson space with respect to itself follows from the identity
 \[
\ad_X^*\iota(Y) = -\iota(\ad_X Y),
 \]
which is a direct consequence of the invariance of $\langle\cdot,\cdot\rangle$ by adjoint action. The remainder is the straightforward application of Theorem~\ref{4.2} to this case.
\end{proof}

\subsection{Banach Lie bialgebras}
Let us recall from \cite{tumpach-bruhat} the notion of Banach Lie bialgebras. 

\begin{definition}\label{Bialgebra_def}
Let $\left(\mathfrak{g}_+, [\cdot, \cdot]_{\g_+}\right)$ be a Banach Lie algebra over the field $\mathbb{K}\in\{\mathbb{R}, \mathbb{C}\}$, and consider a duality pairing $\langle\cdot,\cdot\rangle_{\mathfrak{g}_+, \mathfrak{g}_-}$ 
between $\mathfrak{g}_+$ and a Banach space $\mathfrak{g}_-$.  One says that $\mathfrak{g}_+$ is a \textbf{Banach Lie bialgebra with respect to} $\mathfrak{g}_-$ 
  if 
 \begin{enumerate}
 \item $\mathfrak{g}_+$ acts continuously by coadjoint action on $\mathfrak{g}_-\subset \mathfrak{g}_+^*$~; 
 \item  $\mathfrak{g}_-$ admits a Banach Lie algebra structure  $[\cdot, \cdot]_{\mathfrak{g}_-}:\mathfrak{g}_-\times\mathfrak{g}_-\rightarrow\mathfrak{g}_-$ such that
\begin{equation}\label{cocycle_mitte}
\begin{array}{ll}
\langle[x, y]_{\mathfrak{g}_+},[\alpha, \beta]_{\mathfrak{g}_-}\rangle_{\mathfrak{g}_+, \mathfrak{g}_-}
= &\langle y, [\ad^*_x\alpha, \beta]_{\mathfrak{g}_-}\rangle_{\mathfrak{g}_+, \mathfrak{g}_-} 
+ \langle y, [\alpha, \ad^*_x\beta]_{\mathfrak{g}_-}\rangle_{\mathfrak{g}_+, \mathfrak{g}_-} \\
&- \langle x, [\ad^*_y\alpha, \beta]_{\mathfrak{g}_-}\rangle_{\mathfrak{g}_+, \mathfrak{g}_-} 
- \langle x, [\alpha, \ad^*_y\beta]_{\mathfrak{g}_-} \rangle_{\mathfrak{g}_+, \mathfrak{g}_-} ,
\end{array}
\end{equation}
for all $x, y\in\g_+$ and $\alpha, \beta\in \g_-$.
   \end{enumerate} 
\end{definition}

The following Theorem is a direct consequence of Theorem~2.3 and Theorem~4.9 in \cite{tumpach-bruhat}.
\begin{theorem}\label{manin_equivalence}
Consider two Banach Lie algebras 
 $\left(\mathfrak{g}_+, [\cdot,\cdot]_{\mathfrak{g}_+}\right)$ and  $\left(\mathfrak{g}_-, [\cdot,\cdot]_{\mathfrak{g}_-}\right)$ and denote by $\mathfrak{g}$ the Banach space $\mathfrak{g} = \mathfrak{g}_+\oplus \mathfrak{g}_-$ with norm $\|\cdot\|_{\mathfrak{g}} = \|\cdot\|_{\mathfrak{g}_+}+\|\cdot\|_{\mathfrak{g}_-}$. The following assertions are equivalent:
\begin{itemize}
\item[(1)] $(\mathfrak{g}, \mathfrak{g}_+, \mathfrak{g}_-)$ admits a structure of Manin triple; 
\item[(2)] $\mathfrak{g}_+$ is a Banach Lie--Poisson space and a Banach Lie bialgebra with respect to $\mathfrak{g}_-$; 
\item[(3)] $\mathfrak{g}_-$ is a Banach Lie--Poisson space and a Banach Lie bialgebra with respect to $\mathfrak{g}_+$. 
\end{itemize}
\end{theorem}

\begin{example}{\rm
By Proposition~\ref{triples}, the triple $\left(L^p(\mathcal{H}), \mathfrak{u}_p(\mathcal{H}), \mathfrak{b}_p(\mathcal{H})\right)$ is a Banach Manin triple for $1<p\leq 2$. Under this condition on $p$, it follows from Theorem~\ref{manin_equivalence} that $\mathfrak{u}_p(\mathcal{H})$ is a Banach Lie--Poisson space and a Banach Lie bialgebra with respect to $\mathfrak{b}_p(\mathcal{H})$, and $\mathfrak{b}_p(\mathcal{H})$ is a Banach Lie--Poisson space and a Banach Lie bialgebra with respect to $\mathfrak{u}_p(\mathcal{H})$.}
\end{example}

\begin{example}\label{Iwasawa_Banach_Lie_bialgebra}{\rm 
Let $p$ and $q$ be such that $1<p<\infty$, $1<q<\infty$ and $\frac{1}{p}+\frac{1}{q} = 1$.
Consider the Banach Lie algebra $\mathfrak{u}_p(\mathcal{H})$, and identify its dual Banach space with $\mathfrak{b}_q(\mathcal{H})$ via the pairing given by the imaginary part of the trace. Then  $\mathfrak{u}_p(\mathcal{H})$ is a Banach Lie--Poisson space and a Banach Lie bialgebra with respect to $\mathfrak{b}_q(\mathcal{H})$. We deduce from Theorem~\ref{manin_equivalence} that $\left(\mathfrak{u}_p(\mathcal{H})\oplus \mathfrak{b}_{q}(\mathcal{H}), \mathfrak{u}_p(\mathcal{H}), \mathfrak{b}_{q}(\mathcal{H})\right) $ forms a Banach Manin triple.}
\end{example}

\subsection{Banach Poisson--Lie groups}\label{sec:BPLG}

\begin{definition}
A \textbf{Banach Poisson--Lie group} $G$ is a Banach Lie group  equipped with a generalized Banach Poisson manifold structure such that the group multiplication $m: G \times G \rightarrow G$ is a Poisson map, where $G \times G$ is endowed with the product Poisson structure. Using standard notation, $R_g$ will denote right multiplication by $g\in G$, as well as the induced action on tangent vectors. The induced action in $T^*G$ and $T^{**}G$ will be denoted by $R_g^*$ and $R_g^{**}$. This is not to be confused with the $R$-matrices introduced in next section.
\end{definition}

\begin{proposition}[{\cite[Proposition~5.7]{tumpach-bruhat}}]\label{cela}
A Banach Lie group $G$ endowed with a generalized Banach Poisson structure $(G, \mathbb{F}, \pi)$  is a Banach Poisson--Lie group iff 
\begin{enumerate}
\item $G$ acts continuously on $\mathbb{F}$ by left and right translations;

\item  the map $\Pi: G  \rightarrow  \Lambda^2\mathbb{F}_e^*$ defined by \[ g \mapsto  \Pi(g) := R_{g^{-1}}^{**}\pi_g\]
with
\[
\Pi(g)\left(\alpha, \beta\right) = \pi_g\left(R_{g^{-1}}^*(\alpha), R_{g^{-1}}^*(\beta)\right), \g\in G, \alpha, \beta \in \mathbb{F}_e,
\]
is a $1$-cocycle on $G$ with respect to the coadjoint action $\Ad^{**}$ of $G$ on $\Lambda^2\mathbb{F}_e^*$, i.e.  for any $g, u\in G$,
\begin{equation}\label{pirPoisson}\Pi(gu)  = \Ad_g^{**}\Pi(u) + \Pi(g).\end{equation}
\end{enumerate}
\end{proposition}

\begin{remark}
 Recall that the natural  coadjoint action $\Ad^{**}$ of $G$ on $\Lambda^2\mathbb{F}_e^*$ is defined by
\[
 \Ad_g^{**}\Pi(u)(\alpha, \beta) = \Pi(u)\left(\Ad^*_{g^{-1}}(\alpha), \Ad^*_{g^{-1}}(\beta)\right),
 \]   
 where $g\in G$, and $\alpha,\beta\in \mathbb{F}_e\subset \g^*$.
\end{remark}

\begin{theorem}[{\cite[Theorem~5.11]{tumpach-bruhat}}] \label{PL_bi}
Let $(G_+, \mathbb{F}, \pi)$ be  a Banach Poisson--Lie group. Then the typical fibre $\mathbb F_e$ of the subbundle $\mathbb{F}\subset T^*G_+$ admits a Banach Lie algebra structure denoted as $\g_-$ such that the Lie algebra $\mathfrak{g}_+$ of $G_+$ is a Banach Lie bialgebra with respect to $\mathfrak{g}_-= \mathbb{F}_e$.
\end{theorem}

\begin{remark}
    Given a Banach Poisson--Lie group $(G_+, \mathbb{F}, \pi)$, it follows from \cite[Theorem~5.11]{tumpach-bruhat} that the Lie bracket in $\mathfrak{g}_-:=\mathbb{F}_e$ is given by 
\begin{equation}\label{def_bracket}
[\alpha, \beta]_{\mathfrak{g}_-} := T_e\Pi(\alpha, \beta)\in \mathfrak{g}_-\subset \mathfrak{g}_+^*,\quad  \alpha, \beta \in \mathfrak{g}_-\subset \mathfrak{g}_+^*, 
\end{equation}
where $\Pi:= R_{g^{-1}}^{**}\pi: G_+ \rightarrow \Lambda^2\mathfrak{g}_-^*$, and  $T_e\Pi :\mathfrak{g}_+\rightarrow \Lambda^2\mathfrak{g}_-^*$ denotes the differential of $\Pi$ at the unit element $e\in G_+$.
\end{remark}

\begin{theorem}[{\cite[Theorem~5.13]{tumpach-bruhat}}]\label{PL_LP}
Let $(G_+, \mathbb{F}, \pi)$ be  a Banach Poisson--Lie group. If the map $\pi^{\sharp}: \mathbb{F} \rightarrow \mathbb{F}^*$ defined by $\pi^{\sharp}(\alpha) := \pi(\alpha, \cdot)$ takes values in  $TG_+\subset \mathbb{F}^*$, then $\mathfrak{g}_+$ is a Banach Lie--Poisson space with respect to $\mathfrak{g}_-:=  \mathbb{F}_e$.
\end{theorem}

\begin{corollary}
    Let $(G_+, \mathbb{F}, \pi)$ be  a Banach Poisson--Lie group with Lie algebra $\g_+$ such that $\pi^{\sharp}: \mathbb{F} \rightarrow \mathbb{F}^*$ takes values in $TG_+\subset \mathbb{F}^*$.  Denote by $\g_-$ the fiber $\mathbb{F}_e$ at the unit $e\in G$. Then $\g = \g_+ \oplus \g_-$ is a Banach Manin triple.
\end{corollary}

\subsection{Iwasawa Banach Poisson--Lie groups}
To the Banach Lie algebra $\mathfrak{b}_{p}(\mathcal{H})$ defined by \eqref{b2pm} there is associated the following Banach Lie group:
\[
B_{p}(\mathcal{H}) := \{\alpha\in GL(\mathcal{H})\cap \left(\Id+\mathfrak{b}_{p}(\mathcal{H})\right): \alpha^{-1}\in \Id +\mathfrak{b}_{p}(\mathcal{H}) ~\textrm{and}~\langle n|\alpha|n\rangle\in(0, +\infty), \forall n\in\mathbb{Z}\}.
\]
Both $U_p(\mathcal{H})$ and $B_p(\mathcal{H})$ admit a natural structure of Banach Poisson--Lie groups, that we recall below.

\begin{proposition}[{\cite[Proposition~5.9]{tumpach-bruhat}}]\label{Bp}
For $1<p\leq 2$, consider the Banach Lie group $B_{p}(\mathcal{H})$ with Banach Lie algebra $\mathfrak{b}_p(\mathcal{H})$, and the duality pairing $\langle\cdot,\cdot\rangle_{\mathbb{R}}: \mathfrak{b}_p(\mathcal{H})\times \mathfrak{u}_p(\mathcal{H})\rightarrow \mathbb{R}$ given by the imaginary part of the trace \eqref{imparttrace}.
Consider
\begin{enumerate}
\item $\mathbb{B}_b:= R_{b^{-1}}^*\mathfrak{u}_p(\mathcal{H}) \subset T^*_b B_{p}(\mathcal{H})$, $b\in B_{p}(\mathcal{H})$.
\item $\Pi^{B_p}:B_{p}(\mathcal{H})\rightarrow \Lambda^2\mathfrak{u}_p(\mathcal{H})^*$ defined by 
\begin{equation}\label{PiBp}
\Pi^{B_p}(b)(x_1, x_2) := \langle p_{\mathfrak{b}_p}(b^{-1}x_1b), p_{\mathfrak{u}_p}(b^{-1}x_2 b)\rangle = \im\Tr p_{\mathfrak{b}_p}(b^{-1}x_1b)\left[ p_{\mathfrak{u}_p}(b^{-1}x_2 b) \right], 
\end{equation}
where $b\in B_{p}(\mathcal{H})$ and $x_1, x_2\in \mathfrak{u}_p(\mathcal{H})$.
\item $\pi^{B_{p}}: B_{p} \rightarrow \Lambda^2 T B_{p}(\mathcal{H})$ given by $\pi^{B_{p}}(b) := R_{b}^{**}\Pi^{B_{p}}(b)$.
\end{enumerate}
Then $(B_{p}(\mathcal{H}), \mathbb{B}, \pi^{B_{p}})$ is a Banach Poisson--Lie group.
\end{proposition}

\begin{proposition}[{\cite[Proposition~5.10]{tumpach-bruhat}}]\label{Up}
For $1<p\leq 2$, consider the Banach Lie group  $ {U}_{p}(\mathcal{H})$ with Banach Lie algebra $\mathfrak{u}_p(\mathcal{H})$ and the duality pairing $\langle\cdot,\cdot\rangle_{\mathbb{R}}: \mathfrak{b}_p(\mathcal{H})\times \mathfrak{u}_p(\mathcal{H})\rightarrow \mathbb{R}$ given by the imaginary part of the trace \eqref{imparttrace}. Consider
\begin{enumerate}
\item
$\mathbb{U}_u:= R_{u^{-1}}^*\mathfrak{b}_p(\mathcal{H})\subset T_u^*{U}_{p}(\mathcal{H})$, $u\in {U}_{p}(\mathcal{H})$,
\item $\Pi^{ {U}_{p}}: {U}_{p}(\mathcal{H})\rightarrow \Lambda^2\mathfrak{b}_p(\mathcal{H})^*$ defined by
\begin{equation}\label{PiUp}
\Pi^{ {U}_{p}}(u)(b_1, b_2) := \langle  p_{\mathfrak{u}_p}(u^{-1}b_1u), p_{\mathfrak{b}_p}(u^{-1}b_2 u) \rangle = \im\Tr p_{\mathfrak{u}_p}(u^{-1}b_1u)\left[ p_{\mathfrak{b}_p}(u^{-1}b_2 u)\right],
\end{equation}
where $u \in {U}_{p}(\mathcal{H})$ and $b_1, b_2\in \mathfrak{b}_p(\mathcal{H})$.
\item $\pi^{{U}_{p}}: {U}_{p}(\mathcal{H})\rightarrow \Lambda^2 T{U}_{p}(\mathcal{H})$ given by $\pi^{{U}_{p}}(g) := R_{g}^{**}\Pi^{ {U}_{p}}(g)$.
\end{enumerate}
Then $({U}_{p}(\mathcal{H}), \mathbb{U}, \pi^{{U}_{p}})$ is a Banach Poisson--Lie group. 
\end{proposition}

\begin{remark}
    The tangent bialgebras of the Banach Poisson--Lie groups ${B}_{p}(\mathcal{H})$ and $ {U}_{p}(\mathcal{H})$ defined in Proposition~\ref{Bp} and Proposition~\ref{Up}, are the Banach Lie bialgebra $\mathfrak{b}_p(\mathcal{H})$ and $\mathfrak{u}_p(\mathcal{H})$ in duality, which combine into the Manin triple $(L^p(\mathcal{H}), \mathfrak{u}_{p}(\mathcal{H}), \mathfrak{b}_p(\mathcal{H}))$ given in Proposition~\ref{triples}.
\end{remark}

\section{\texorpdfstring{$R$-matrices on a Banach Lie algebra}{R-matrices on a Banach Lie algebra}}\label{sec:R}

\subsection{Definition of \texorpdfstring{$R$}{R}-matrices in the Banach context}
Let us recall the definition of $R$-matrices adapted to the Banach context, and basic facts around this notion (see e.g. \cite{belavin82,semenov83,kosmann88,adler} for the finite-dimensional case).

\begin{definition}\label{def:R}
Let $\mathfrak{g}$ be a Banach Lie algebra. A \textbf{classical $R$-matrix} is a bounded linear operator $R:\mathfrak{g}\rightarrow \mathfrak{g}$ such that 
the skew-symmetric continuous bilinear map defined by
\begin{equation}\label{R}
[x, y]_{R} = \frac{1}{2}\left([Rx, y] + [x, Ry]\right),\quad\quad \forall x, y\in\mathfrak{g},
\end{equation}
is a Lie bracket on $\mathfrak{g}$, called  the \textbf{$R$-bracket}. 
The pair $(\mathfrak g, R)$ is called a \textbf{double Banach Lie algebra}. The Banach Lie algebra $\g$ with the bracket $[\cdot,\cdot]_R$ will be denoted $\g_R$.
\end{definition}

\begin{remark}
    For an arbitrary Banach Lie--Poisson space $\bg$ with respect to a Banach Lie-algebra $\g$ endowed with a classical $R$-matrix $R$, it is not guaranteed that the bracket $[\cdot,\cdot]_R$  leads to a Poisson structure on $\bg$ in the sense of Definition~\ref{def:poissonmanifold}. Namely the condition \eqref{blp-condition} may not hold in general for the coadjoint representation related to $[\cdot,\cdot]_R$.

    In the case that the condition \eqref{blp-condition} holds, we will denote by $\pb_R$ the Lie--Poisson bracket on the algebra $\mathcal{A}$ of smooth functions on $\bg$ with derivatives in $\g$ associated with the bracket $[\cdot, \cdot]_R$.
\end{remark}

\begin{proposition}\label{RLiePoisson}
Let $\bg$ be a Banach Lie--Poisson space with respect to $\g$ and let $R$ be a classical $R$-matrix $R$ on $\g$. If the dual map $R^*:\g^*\to\g^*$ preserves $\bg$
\[ R^*\bg\subset \bg,\] 
then $\bg$ is also a Banach Lie--Poisson space with respect to the Banach Lie algebra $(\g, [\cdot, \cdot]_R)$.
\end{proposition}
\begin{proof}
    By Definition~\ref{def:R}, the coadjoint representation related to the Lie bracket $[\cdot,\cdot]_R$ is
    \begin{equation}\label{coad_R}
    (\ad_R^*)_x = \frac12\big(\ad^*_{Rx} + R^*\ad^*_x\big),
    \end{equation}
    where $x\in\g$.
        From Definition~\ref{def:blp}, for $\bg$ to be a Banach Lie--Poisson space with respect to the Banach Lie algebra $\g_R$, we need the map $(\ad_R^*)_x$ to take values in $\bg$ for all $x\in\g$. Since we assumed that $\bg$ is a Banach Lie--Poisson space with respect to $\g$, both $\ad^*_{Rx}$ and $\ad^*_x$ preserve $\bg$. Thus a sufficient condition to get a Banach Lie--Poisson structure on $\bg$ with respect to $\left(\g, [\cdot, \cdot]_R\right)$ is for $R^*$ to preserve $\bg$ as well.
\end{proof}

There is a subclass of $R$-matrices which are solutions of the so-called modified classical Yang--Baxter equation.
\begin{proposition}
Let $\mathfrak{g}$ be a Banach Lie algebra. A bounded linear operator $R:\mathfrak{g}\rightarrow \mathfrak{g}$, which satisfies the following equation, known as the \textit{modified classical Yang--Baxter equation} (mCYBE): 
\begin{equation}\label{mCYBE}
[Rx, Ry] = R\left([Rx, y] + [x, Ry]\right) - [x, y], \quad \forall x, y \in\mathfrak{g},
\end{equation}
is a classical $R$-matrix.
\end{proposition}
\begin{proof}
One has 
\[
\begin{array}{ll}
4\left[[x, y]_R, z\right]_R &= 2\left[[Rx, y] + [x, Ry], z\right]_R = \left[R\left([Rx, y] + [x, Ry]\right), z\right] +\left[[Rx, y] + [x, Ry],Rz\right]\\
& = \left[[Rx, Ry] + [x, y], z\right] + \left[[Rx, y] + [x, Ry],Rz\right] \\
& = \left[[x, y], z\right] + \left[[Rx, Ry], z\right] + \left[[Rx, y],  Rz\right] + \left[[x, Ry], Rz\right],
\end{array}
\]
and the Jacobi identity of $[\cdot, \cdot]_R$ follows from the Jacobi identity satisfied by $[\cdot, \cdot]$.
\end{proof}

\begin{proposition}\label{prop:Rpm}
Given a $R$-matrix $R$ satisfying the modified classical Yang--Baxter equation \eqref{mCYBE} on a Banach Lie algebra $\mathfrak{g}$, the maps $R_\pm = \frac{1}{2}(R\pm \id)$ are Lie algebra homomorphisms from $\left(\mathfrak{g}, [\cdot, \cdot]_R\right)$ into $\left(\mathfrak{g}, [\cdot, \cdot]\right)$, where $\id$ denotes the identity map.
\end{proposition}

\begin{proof}
By direct calculation one gets
\[
\begin{array}{ll}
R_+\left([x, y]_R\right) &=\frac{1}{2} R_+\left( [Rx, y] + [x, Ry]\right) = \frac{1}{4} R\left( [Rx, y] + [x, Ry]\right) + \frac{1}{4}\left( [Rx, y] + [x, Ry]\right)\\
& = \frac{1}{4}[Rx, Ry] +\frac{1}{4}[x, y] +\frac{1}{4}[Rx, y] + \frac{1}{4}[x, Ry]\\
& = \left[\frac{1}{2}Rx +\frac{1}{2}x, \frac{1}{2} Ry + \frac{1}{2} y\right]\\
& = \left[\frac{1}{2}\left(R + \id\right)x, \frac{1}{2}\left(R + \id\right)y\right] = \left[R_+ x, R_+ y\right]
\end{array}
\]
and similarly for $R_-$.
\end{proof}

\subsection{\texorpdfstring{$R$}{R}-matrices associated with the sum of Banach Lie subalgebras}
\hfill
We shall present now a widely used method of obtaining examples of classical $R$-matrices, namely when the Lie algebra $\mathfrak{g}$ admits a Banach decomposition into the direct sum of two closed Lie Banach subalgebras: $\mathfrak{g} = \mathfrak{g}_+\oplus\mathfrak{g}_-$. This situation can be traced back under different names in the literature: under the name ``twilled extension'' or ``twilled Lie algebra''  \cite{KS89, kosmann88}, or ``algèbre de Lie bicroisée'' \cite{aminou1988groupes}, under the name ``bicrossproduct Lie algebra'' \cite{majid1988}, or under the name ``double Lie algebra'' in \cite{lu-weinstein90}, which differs from the more general definition of double Lie algebra given in Definition~\ref{def:R}.

\begin{proposition}\label{example_R}
Assume that the Banach Lie algebra $\mathfrak{g}$ admits a Banach decomposition into the direct sum of two closed Lie Banach subalgebras: $\mathfrak{g} = \mathfrak{g}_+\oplus\mathfrak{g}_-$. Set $R = p_+ - p_-$, where $p_{\pm}$ is the projection onto $\mathfrak{g}_{\pm}$ with respect to the previous decomposition. 
   Then $R$ is a classical $R$-matrix which satisfies the  modified classical Yang–Baxter
equation \eqref{mCYBE}. The $R$-bracket on $\mathfrak{g}$ reads
\begin{equation}\label{R-bracket}
[x, y]_R = [x_+, y_+] - [x_-, y_-],
\end{equation}
with $x_\pm = p_{\pm}(x)$ and $y_{\pm} = p_{\pm}(y)$. Note that in this case the Lie algebra homomorphisms $R_\pm$ are exactly $\pm p_\pm$.
\end{proposition}
\begin{proof}
    It is straightforward that $[\cdot,\cdot]_R$ is a Lie bracket since it is the Lie bracket of the direct sum of the Lie algebras $\g_+$ and $\g_-$, where the bracket on $\g_-$ is minus the restriction of $[\cdot, \cdot]$ to $\g_-$. 
    The fact that $R$ satisfies the modified classical Yang--Baxter equation \eqref{mCYBE} follows from e.g. \cite[Proposition~5]{semenov83} or \cite[Lemma~4.34]{adler}.
\end{proof}

More generally, one have the following example of $R$-matrix.
\begin{proposition}
    Assume that the Banach Lie algebra $\mathfrak{g}$ admits a Banach decomposition into a direct sum $\g = \g_+\oplus \g_0\oplus \g_-$, where
    \begin{itemize}
        \item $\g_+$ and $\g_-$ are Banach Lie subalgebras of $\g$;
        \item $\g_0$ is an abelian Banach Lie subalgebra of $\g$;
        \item $\g_0$ normalizes $\g_+$ and $\g_-$, i.e.
  $  [\g_0, \g_+] \subset \g_+ \textrm{ and } [\g_0, \g_-] \subset \g_-.
        $
    \end{itemize}

Denote by $p_+(x) = x_+$, $p_0(x) = x_0$ and $p_-(x) = x_-$ the projections of $x\in\g$ onto $\g_+$, $\g_0$ and $\g_-$ respectively.
Then $R = p_+ - p_-$ is a classical $R$-matrix, which satisfies the  modified classical Yang–Baxter
equation \eqref{mCYBE}. The $R$-bracket on $\mathfrak{g}$ reads
\begin{equation}\label{R-bracket-proj}
[x, y]_R = [x_+, y_+]  - [x_-, y_-]  +\frac{1}{2}[x_+-x_-, y_0]   +\frac{1}{2}[x_0, y_+-y_-] .
\end{equation}
\end{proposition}
\begin{proof}
    One has
    \[
    \begin{array}{ll}
[Rx, Ry]  + [x, y]  &= [x_+-x_-, y_+-y_-]  + [x_++x_0+ x_-, y_++y_0+y_-] \\ &  = 2[x_+, y_+]  + 2[x_-, y_-]  + [x_0, y_+]  +  [x_0, y_-]  + [x_+, y_0]  + [x_-, y_0] .
\end{array}
    \]
    On the other hand
    \[
    \begin{array}{ll}
R[Rx, y]  + R[x, Ry]  &= R[x_+-x_-, y]  + R[x, y_+-y_-] \\   
& =  [x_+, y_++y_0]  + R[x_+, y_-]  +[x_-, y_- + y_0]  -R[x_-, y_+]  \\
& + [x_+ + x_0, y_+]  + R[x_-, y_+] + [x_-+x_0, y_-] - R[x_+, y_-]\\
& = 2[x_+, y_+]  + 2[x_-, y_-]  +[x_+, y_0]   + [x_-, y_0]  + [x_0, y_+]  +  [x_0, y_-] ,
\end{array}
    \]
    hence $R$ satisfies the  modified classical Yang–Baxter
equation \eqref{mCYBE}. The corresponding bracket reads:
\[
\begin{array}{ll}
[x, y]_R & = \frac{1}{2}[x_+-x_-, y_+ + y_0 + y_-]  + \frac{1}{2}[x_+ + x_0 + x_-, y_+-y_-] \\
& = [x_+, y_+]  - [x_-, y_-]  +\frac{1}{2}[x_+-x_-, y_0]   +\frac{1}{2}[x_0, y_+-y_-] . 
\end{array}
\]
\end{proof}

\subsection{Functions in involution for Lie--Poisson brackets given by \texorpdfstring{$R$}{R}-matrices}\label{sec:Involutivity}

In this section we review the theory that leads to functions in involution for the Lie--Poisson bracket associated to an $R$-matrix. We refer the reader to the Adler--Kostant--Symes (AKS) Theorem in the finite-dimensional setting \cite[Chapter~4.4]{adler} or \cite[Chapter~12.2]{laurent-gengoux}. Here we present a simplified version first (with $\varepsilon = 0$) of the AKS theorem (see Theorem~\ref{AKS}) but adapted to the infinite-dimensional setting (subsection~\ref{epsilon0}). We then construct Banach Lie--Poisson spaces from a Banach Lie--Poisson $\bg$ with respect to a Banach Lie algebra that admits a decomposition into  the sum of two Lie subalgebras (subsection~\ref{subalgebras}). We use this construction to give a generalization of the AKS theorem to the Banach setting (for any $\varepsilon\in\bg$) which is adapted to arbitrary duality pairings between  Banach spaces $\bg$ and  Banach Lie algebras $\g$ (subsection~\ref{epsilonnon0}). Finally in subsection~\ref{factorization}, we present the solution of the Hamiltonian flows using solutions of the corresponding factorization problem.
Note that in the infinite-dimensional setting, not every Banach Lie algebra can be integrated to a Banach Lie group (see e.g. \cite{vanest64,glockner-neeb03,beltita-pelletier2025}).
For this reason, we first present involutivity theorems for functions that are invariant by the coadjoint action of a Lie algebra $\g$ as opposed to functions that are invariant by the coadjoint action of a Lie group $G$ integrating $\g$. Note that if $G$ is a Banach Lie group with Lie algebra $\g$, then any $\Ad^*_G$-invariant function $F$ is also $\ad^*_\g$-invariant.

\begin{lemma}\label{adstar}
    By definition, any function $F\in \Cp^\infty(\bg)$ is invariant by coadjoint action if and only if
\begin{equation}\label{coadjinv}
    D_\mu F\left(\ad^*_X\mu\right) = 0  \quad \forall X\in\g, \forall \mu \in \bg.
    \end{equation}
    This condition is equivalent to
    \begin{equation}\label{hinv}
    \ad^*_{D_\mu F}\mu (X) =  0\quad \forall X\in\g, \forall \mu \in \bg.
    \end{equation}
\end{lemma}
\begin{proof}
One has
\[
     \ad^*_{D_\mu F}\mu (X) = \langle \mu, [D_\mu F, X]_\g\rangle = -\langle \ad^*_X \mu, D_\mu F\rangle = - D_\mu F\left(\ad^*_X \mu\right) = 0.
     \]
\end{proof}

\subsubsection{Involutivity theorem (simplified version of AKS theorem with $\varepsilon =0$)}\label{epsilon0}
Recall that, for a Banach Lie--Poisson space $\bg$ with respect to a Banach Lie algebra $\g$, $\mathcal{A}$ denotes the unital subalgebra of $\Cp^{\infty}(\bg)$ consisting of all  functions with differentials in $\g$, see \eqref{Ag}.

Let us first present a simplified version of the AKS theorem (with argument shift $\varepsilon = 0$), but suitable to our Banach setting of generalized Poisson structures. We refer the reader to \cite[Theorem 1]{semenov83} or \cite[Theorem 4.36]{adler} for the original versions of the following theorem in finite dimensional setting. Note that in points (1) and (2) of Theorem~\ref{AKS}, we don't assume that the $R$-matrix comes from the decomposition of the Banach Lie algebra $\g$.

\begin{theorem}\label{AKS}
    Let $\bg$ be a Banach Lie--Poisson space with respect to a Banach Lie algebra  $\g$, and $R$ a classical $R$-matrix $R$ on $\g$ such that $\bg$ is also a Banach Lie--Poisson space with respect to $\g_R$.
       Then we have:
    \begin{enumerate}
    \item $\{F,G\}_R = 0$ for any functions $F,G\in \mathcal{A}$ which are $\ad^*_\g$-invariant \eqref{coadjinv}.
    \item The Hamiltonian vector field generated by an $\ad^*_\g$-invariant function $F\in\mathcal A$ with respect to the Poisson bracket $\pb_R$ assumes the form
    \[X_F(\mu) = \tfrac12 \ad^*_{R D_\mu F }\mu.\]
    \item If moreover  $R$ is the $R$-matrix associated with a decomposition $\g = \g_+\oplus \g_-$ into the sum of Banach Lie subalgebras, then the Hamiltonian vector field generated by an $\ad^*_\g$-invariant function $F\in \mathcal{A}$  with respect to the Poisson bracket $\pb_R$  reads
    \begin{equation}\label{hamvec}
    X_F(\mu) = \ad^*_{(D_\mu F)_+}\mu = - \ad^*_{(D_\mu F)_-}\mu, 
    \end{equation}
    for $\mu \in \bg$, where $(D_\mu F)_\pm = p_\pm(D_\mu F)$.
\end{enumerate}
\end{theorem}

\begin{proof}$\;$
\begin{enumerate} \item
    By Proposition~\ref{RLiePoisson}, $\bg$ is a Lie--Poisson space with respect to $\g$ for the Lie bracket $[\cdot, \cdot]_R$. Denote by $\langle\cdot, \cdot\rangle$ the duality pairing between $\bg$ and $\g$. Using the definition of the $R$-bracket~\eqref{R}, for $\mu \in \bg$ one has
    \[
    \begin{array}{ll}
    \{F,G\}_R(\mu) &= \langle \mu, [D_\mu F, D_\mu G]_R\rangle = \frac{1}{2} \langle \mu, [R D_\mu F, D_\mu G]_{\g}\rangle + \frac{1}{2} \langle \mu, [ D_\mu F, R D_\mu G]_{\g}\rangle\\
    & = -\frac{1}{2} \langle \ad^*_{R D_\mu F}\mu, D_\mu G\rangle + \frac{1}{2} \langle  \ad^*_{R D_\mu G}\mu, D_\mu F\rangle\\
    & = -\frac{1}{2}  D_\mu G\left(\ad^*_{R D_\mu F}\mu\right) + \frac{1}{2} D_\mu F\left(\ad^*_{R D_\mu G}\mu\right) = 0
    \end{array}
    \]
    by equation~\eqref{coadjinv}.
    \item For a general function $H\in \mathcal{A}$ and $F$ invariant by coadjoint action, one has
    \[
    \begin{array}{ll}
    X_F(\mu)(H) = \{H,F\}_R(\mu) = \frac{1}{2}  D_\mu H\left(\ad^*_{R D_\mu F}\mu\right).
    \end{array}
    \]
    Hence \[X_F(\mu) = \tfrac12 \ad^*_{R D_\mu F }\mu.\]

\item In this case 
\[
R D_\mu F = (D_\mu F)_+ - (D_\mu F)_- =  2(D_\mu F)_+ - D_\mu F = D_\mu F - 2(D_\mu F)_-.
\]
Moreover, for any $X\in \g$ and any function $F\in \mathcal A$ invariant by coadjoint action, one has
\[
    \ad^*_{D_\mu F}\mu (X) = \langle \mu, [D_\mu F, X]_\g\rangle = -\langle \ad^*_X \mu, D_\mu F\rangle = - D_\mu F\left(\ad^*_X \mu\right) = 0.
\]
Thus using \eqref{coad_R} we get $(\ad^*_R)_{D_\mu F} = \frac12 (\ad^*_{R D_\mu F} + R^*\ad^*_{D_\mu F}) = \ad^*_{(D_\mu F)_+} = - \ad^*_{(D_\mu F)_-}$. 
 \end{enumerate}
\end{proof}

\subsubsection{Lie--Poisson structures induced by a decomposition $\g = \g_+\oplus \g_-$}\label{subalgebras}
As mentioned in Remark~\ref{existence_predual}, a closed subspace of a Banach space admitting a predual might not admit a predual. However when $\bg$ is a Banach Lie--Poisson space with respect to a Banach Lie algebra $\g$ which admits a decomposition into the sum of two subalgebras $\g_+$ and $\g_-$, more results can be formulated (see Proposition~\ref{prop_more_results} below).

\begin{remark}\label{rem:annih}
Denote by
$\g_\pm^0\subset \g^*$ the annihilator of $\g_\pm$
  \[
  \g_\pm^0 = \{ f\in\g^*, \langle f, X\rangle_{\g^*, \mathfrak{g}} = 0\quad \forall X\in \g_\pm\}.
  \]
Let us consider for now the projections $p_\pm$ as maps from $\g$ to $\g_\pm$.
From the decomposition $\g = \g_+\oplus \g_-$, it follows that the dual maps
     \[
\iota_+ := p_+^*: \g_+^*\rightarrow \g^*
     \]
and 
     \[
\iota_- := p_-^*: \g_-^*\rightarrow \g^*
     \]
are continuous and injective, with range $\g_-^0$ and $\g_+^0$ respectively. Consequently, $\g_+^*\simeq \g_-^0$ and $\g_-^*\simeq \g_+^0$. Thus we have the decomposition
  \[ \g^* = \g_-^0 \oplus \g_+^0 = \g_+^* \oplus \g_-^*. \]
  Note that 
  \[
\iota_\pm^* := (p_\pm^*)^*: \g^{**}\rightarrow \g_\pm^{**}
     \]
     restricts to $p_\pm$ on $\g\subset \g^{**}$.
\end{remark}

\begin{proposition}\label{prop_more_results}
  Let $\bg$ be a Banach Lie--Poisson space with respect to a Banach space $\g$ admitting a decomposition $\g = \g_+\oplus \g_-$ into the sum of Banach Lie subalgebras.  
  Suppose that $\bg$ has a decomposition
  \begin{equation}\label{decbg}
\bg = \left(\g_-^0\cap\bg\right) \oplus \left(\g_+^0\cap\bg\right)
\end{equation}
  into the sum of two closed subspaces. Then $\g_\pm^0\cap\bg$  is a Banach--Lie Poisson space with respect to $\g_\mp$.
\end{proposition}

\begin{proof}
     From Remark~\ref{rem:annih} we conclude that $\g_-^0\cap\bg \subset \g_+^*$ and $\g_+^0\cap\bg \subset \g_-^*$. By hypothesis, $\g_-^0\cap\bg$ and $\g_+^0\cap\bg$ are closed complementary subspaces of $\bg$. Endowed with the topology of $\bg$, they are therefore Banach spaces. In order for $\g_\pm^0\cap\bg$ to be a Banach--Lie Poisson space with respect to $\g_\mp$, one needs to check that $\g_\mp$ acts continuously on $\g_\pm^0\cap\bg$ by coadjoint action, i.e.
\[
\ad^*_x b \in\g_\pm^0\cap\bg
\]
for all $x \in \g_\mp$ and $b\in\g_\pm^0\cap\bg$, and $\ad^*:\g_\mp\times\g_\pm^0\cap\bg\rightarrow \g_\pm^0\cap\bg$ is continuous. The fact that the coadjoint action of $\g_\mp$ preserves $\g_\pm^0\cap\bg$ follows from the fact that the coadjoint action of $\g_\mp$ preserves $\g_\mp^*\simeq \g_\pm^0$, and also $\bg$, since $\bg$ is a Banach--Lie Poisson space with respect to $\g$. The continuity of the coadjoint actions follows from the continuity of the coadjoint action of $\g$ on $\bg$ and of the projections.
\end{proof}

\begin{lemma}
  The decomposition~\eqref{decbg} exists exactly when the $R$-matrix $R = p_+-p_-$ preserves the space $\bg$ and $R^*$ is continuous on $\bg$.
\end{lemma}

\begin{proof}
    For $R = p_+-p_-$, one has $R^* = p_+^* - p_-^*:\g^*\rightarrow \g^*$. Note that $R^* + \id_{\g^*} = p_+^* - p_-^* + p_+^* + p_-^* = 2 p^*_+$. 

    Suppose that $R = p_+-p_-$ satisfies $R^*\bg\subset \bg$ and $R^*$ is continuous on $\bg$. 
    Since $p_+^* =\frac{1}{2}(R^* + \id_{\g^*})$, the condition $R^*\bg \subset \bg$ implies  $p_+^*\bg \subset \bg$ and $p^*_-\bg \subset\bg$. The continuity of $R^*:\bg \rightarrow \bg$, then implies the continuity of 
    $p_{+|\bg}^*:\bg\rightarrow \bg$ and $p_{-|\bg}^*:\bg\rightarrow \bg$.
    Consequently, using Remark~\ref{rem:annih}, one has a decomposition~\eqref{decbg}.
    Moreover since $\g_-^0\cap\bg =  \Ker(p_{-|\bg}^*)$ and $\g_+^0\cap\bg  = \Ker(p_{+|\bg}^*)$, they are closed subspaces of $\bg$.
    
    Reciprocally, suppose that we have a decomposition~\eqref{decbg} into closed subspaces. Then $p^*_+(\bg) \subset \bg$ and $p^*_-(\bg)\subset \bg$. Consequently $R^* = p^*_+-p^*_-$ preserves $\bg$ and is continuous on $\bg$.  
\end{proof}

\begin{proposition}\label{proppoisson}
    Let $\bg$ be a Banach Lie--Poisson space with respect to a Banach Lie algebra $\g$ admitting a decomposition
     $\g = \g_+\oplus \g_-$ into the sum of Banach Lie subalgebras. Consider the $R$-matrix $R = p_+ - p_-$. Suppose that $R^*$
     preserves $\bg$ and is continuous on $\bg$. 
     Denote by $\pb_\pm$ the Lie--Poisson bracket on $\g_\pm^*\cap\bg$. Then
\[
\iota_+:= p_+^*: \left(\g_+^*\cap\bg, \pb_+\right) \rightarrow \left(\bg, \pb_R\right)
\]
is a Poisson map and 
\[
\iota_-:= p_-^*: \left(\g_-^*\cap\bg, \pb_-\right) \rightarrow \left(\bg, \pb_R\right)
\]
is an anti-Poisson map.
     \end{proposition}

     \begin{proof}
Let $\mathcal{A}_\pm$ be the unital subalgebra of $\Cp^{\infty}(\g_\pm^*\cap\bg)$ consisting of all  functions with differentials in $\g_\pm$:
\[
 \mathcal{A}_\pm :=\{ F_\pm \in \Cp^\infty(\g_\pm^*\cap\bg): D_xF_\pm \in \g_\pm\subset (\g_\pm^*\cap\bg)^* \textrm{ for any }x\in\bg\}.
\]
The generalized Poisson bracket of two functions $F_\pm, G_\pm\in\mathcal{A}_\pm$ takes the form
\begin{equation}
\{F_\pm, G_\pm\}_\pm(x) := \left\langle x, [D_xF_\pm, D_xG_\pm]_{\g_\pm}\right\rangle_{\bg, \g},
\end{equation}
   for $x\in \g_+^*\cap\bg$ and $F_\pm, G_\pm\in\mathcal{A}_\pm$ (see \eqref{pipoisson-blp}). On the other hand, on the subalgebra $\mathcal{A}$ of $\Cp^{\infty}(\bg)$ consisting of all  functions with differentials in $\g$
the generalized Lie--Poisson bracket corresponding to the $R$-bracket $[\cdot,\cdot]_R$ reads  
\begin{equation}
\{H, K\}_R(x) = \left\langle x, [D_xH, D_xK]_R\right\rangle_{\bg, \g},
\end{equation}
where $H, K$ in $\mathcal{A}$ and $x\in \bg$. Note that for $X\in\g$, considered as a linear function on $\bg$, and $x\in \g_\pm^*\cap\bg$, 
\[
\langle \iota_\pm^*(X), x\rangle_{\bg^*,\bg} = \langle x, p_\pm(X)\rangle_{\bg, \g},
\]
hence $\iota_\pm^*$ restrict to $p_\pm$ on $\g\subset \bg^*$. 
It follows that  for $H, K\in \mathcal{A}$, the functions $F_\pm = H\circ \iota_\pm$ and $G_\pm = K\circ \iota_\pm$ belong to the subalgebras $\mathcal{A}_\pm$, and their differentials at $x \in \g_\pm^*\cap\bg$ are respectively equal to
\[
D_xF_\pm = \iota_\pm^*(D_xH) = (D_xH)_\pm\textrm{  and  } D_xG_\pm = \iota_\pm^*(D_xK) = (D_xK)_\pm.
\]
By definition of the Lie--Poisson brackets on $\mathcal{A}_\pm$, for $x\in \g_\pm^*\cap\bg$ and $F_\pm =  H\circ \iota_\pm$, $G_\pm =  K\circ \iota_\pm$, one has
\[
\{F_\pm, G_\pm\}_\pm(x) = \{ H\circ \iota_\pm,  K\circ \iota_\pm\}_\pm(x)  = \left\langle x, [(D_xH)_\pm, (D_xK)_\pm]_{\g_\pm}\right\rangle_{\bg, \g} .
\]
On the other hand, for $x \in \g_\pm^*\cap\bg$,
\begin{align*}
\iota_\pm^*\{H, K\}_R(x) &= \{H, K\}_R(\iota_\pm(x)) = \{H, K\}_R(p^*_\pm(x)) = \langle x, p_{\pm}([D_xH, D_x K]_R)\rangle\\
& = \langle x, p_{\pm}([(D_xH)_+, (D_x K)_+] - [(D_xH)_-, (D_x K)_-])\rangle.
\end{align*}
It follows that 
\[
\iota_+^*\{H, K\}_R(x)  = \langle x, [(D_xH)_+, (D_x K)_+]_{\g_+}\rangle  = \{ H\circ \iota_+,  K\circ \iota_+\}_\pm(x)
\]
and 
\[
\iota_-^*\{H, K\}_R(x)  = -\langle x,  [(D_xH)_-, (D_x K)_{-}]_{\g_-}\rangle = -\{ H\circ \iota_-, K\circ \iota_-\}_\pm(x).
\]
     \end{proof}

\subsubsection{Involutivity theorem ($\varepsilon$ version of AKS theorem)}\label{epsilonnon0}
In this section,  $\bg$ is a Banach Lie--Poisson space with respect to a Banach Lie algebra $\g$ admitting a decomposition $\g = \g_+\oplus \g_-$ into the sum of Banach Lie subalgebras. Consider  the $R$-matrix $R = p_+ - p_-$, which gives rise to another Lie--Poisson bracket $\pb_R$ defined also for functions in the subalgebra $\mathcal{A}$ of $\Cp^{\infty}(\bg)$.
For $\varepsilon\in\bg$ and $H\in\mathcal{A}$, let us introduce the following functions
\[
\begin{array}{l}
\tilde{H}_\varepsilon : \bg \rightarrow \mathbb{C}, x\mapsto H(\varepsilon+x)\\
H_\varepsilon:= \tilde{H}_\varepsilon\circ \iota_+: \g_+^*\cap\bg \rightarrow \mathbb{C}.
\end{array}
\]
Note that $\tilde{H}_\varepsilon\in\mathcal{A}$ and $H_\varepsilon = \tilde{H}_\varepsilon\circ \iota_+$ belongs to $\mathcal{A}_+$. 

\begin{theorem}\label{thmepsilonnon0}
    Suppose that $\varepsilon\in\bg$ satisfies
    \begin{equation}
\langle \varepsilon, [\g_+, \g_+]\rangle_{\bg, \g} = 0 = \langle \varepsilon, [\g_-, \g_-]\rangle_{\bg, \g}.
    \end{equation}
    Then 
    \begin{enumerate}
        \item $\{ H_\varepsilon, K_\varepsilon\}_+ = 0$ for any functions $H,K\in \mathcal{A}$ which are $\ad^*_\g$-invariant \eqref{coadjinv}.
        \item Consider an $\ad^*_\g$-invariant function $H\in\mathcal A$. Then the Hamiltonian vector field $X_{H_\varepsilon} := \{\cdot, H_\varepsilon\}_+$ is given at $x \in \g_+^*\cap\bg$ by
        \begin{equation}\label{XHepsilon}
X_{H_\varepsilon}(x) = \frac{1}{2}\ad^*_{R D_{x+ \varepsilon} H}(x+\varepsilon) = \pm \ad^*_{(D_{x+\varepsilon}H)_\pm}(x + \varepsilon).
        \end{equation}
    \end{enumerate}
\end{theorem}

\begin{proof}
\begin{enumerate}
    \item 
    Since by Proposition~\ref{proppoisson}, \[
\iota_+:= p_+^*: \left(\g_+^*\cap\bg, \pb_+\right) \rightarrow \left(\bg, \pb_R\right)
\]
is a Poisson map, one has
\[
\{ H_\varepsilon, K_\varepsilon\}_+(x) = \{\tilde{H}_\varepsilon\circ \iota_+, \tilde{K}_\varepsilon\circ \iota_+\}_+(x) = \{\tilde{H}_\varepsilon, \tilde{K}_\varepsilon\}_R(\iota_+(x)) = \{\tilde{H}_\varepsilon, \tilde{K}_\varepsilon\}_R(x),
\]
where $\iota_+(x) = x $ for $x \in \g^0_-\cap\bg$.
In order to prove (1), it is therefore sufficient to prove that 
\[
\{\tilde{H}_\varepsilon, \tilde{K}_\varepsilon\}_R(x) = 0
\]
for any $x\in \g^0_-\cap\bg = \g^*_+\cap\bg$. One has
\begin{align*}
    \{\tilde{H}_\varepsilon, \tilde{K}_\varepsilon\}_R(x) &= \langle x, [D_x\tilde{H}_\varepsilon, D_x\tilde{K}_\varepsilon]_R\rangle\\
    & = \langle x, [(D_x\tilde{H}_\varepsilon)_+, (D_x\tilde{K}_\varepsilon)_+] - [(D_x\tilde{H}_\varepsilon)_-, (D_x\tilde{K}_\varepsilon)_-]\rangle,\\
    & = \langle x+\varepsilon, [(D_x\tilde{H}_\varepsilon)_+, (D_x\tilde{K}_\varepsilon)_+] - [(D_x\tilde{H}_\varepsilon)_-, (D_x\tilde{K}_\varepsilon)_-]\rangle,
\end{align*}
where we have used the condition on $\varepsilon$. 
Since 
$D_x \tilde{H}_\varepsilon = D_{\varepsilon+x} H$, 
one has
\begin{align*}
    \{\tilde{H}_\varepsilon, \tilde{K}_\varepsilon\}_R(x) &=  \langle x+\varepsilon, [D_{x+\varepsilon}H, D_{x+\varepsilon}K]_R \rangle\\
    & = \{ H, K\}_R(x+\varepsilon) = 0
\end{align*}
by Theorem~\ref{AKS}(1) applied to the $\ad^*_\g$-invariant functions $H$ and $K$.

\item We have seen that
\[
\{ H_\varepsilon, K_\varepsilon\}_+(x) = \{\tilde{H}_\varepsilon, \tilde{K}_\varepsilon\}_R(\iota_+(x)) 
     = \{ H, K\}_R(\iota_+(x)+\varepsilon),
\]
hence 
\begin{align*}
    \{ H_\varepsilon, K_\varepsilon\}_+(x) &  = -X_H(D_{\iota_+(x)+\varepsilon}K) = -X_H(D_{\iota_+(x)}\tilde{K}_\varepsilon)\\& = -\langle X_H(\iota_+(x)+\varepsilon), D_{\iota_+(x)}\tilde{K}_\varepsilon\rangle_{\bg,\g} 
    \end{align*}
On the other hand,
\begin{align*}
    \{ H_\varepsilon, K_\varepsilon\}_+(x) &= -X_{H_\varepsilon}(D_x K_\varepsilon)  = -X_{H_\varepsilon}(D_x (\tilde{K}_\varepsilon\circ \iota_+)) = -X_{H_\varepsilon}( \iota_+^*D_{\iota_+(x)}\tilde{K}_\varepsilon)\\&  = -\langle i_+(X_{H_\varepsilon}(x)), D_{\iota_+(x)}\tilde{K}_\varepsilon\rangle_{\bg, \g}.  
    \end{align*}
     Recall that $\mathcal{A}$ is an algebra of functions on a linear space $\bg$, hence linear functionals in $\g$ are globally defined on $\bg$. Consequently $D_{\iota_+(x)}\tilde{K}_\varepsilon$ spans $\g$ when $K$ runs over $\mathcal{A}$, and comparing the two expressions of $\{ H_\varepsilon, K_\varepsilon\}_+(x)$ leads to
    \[
\iota_+(X_{H_\varepsilon}(x)) = X_H(\iota_+(x)+\varepsilon).
    \]
    The formulas for $X_{H_\varepsilon}$ then follow from  Theorem~\ref{AKS}(2) applied to the $\ad^*_\g$-invariant function $H$.

\end{enumerate}

\end{proof}
\subsubsection{Integral curves of Hamiltonian vector fields via solutions of the factorization problem}\label{factorization}
In this section, we suppose that $\bg$ is a Banach Lie--Poisson space with respect to a Banach Lie algebra $\g$ which decomposes as $\g = \g_+\oplus \g_-$, and
     that there exist a Banach Lie group $G$, with Lie algebra $\g$, and  two Banach Lie subgroups $G_+$ and $G_-$ of $G$ with Lie algebras $\g_+$ and $\g_-$ respectively. 
We will refer to the factorization problem as the following question. 

\textbf{Factorization problem:}\label{factpb}
Given $X\in \g = \g_+ \oplus \g_-$, find 
 a smooth curve $g_+(t)\in G_+$ and a smooth curve $g_-(t)\in G_-$ solving 
    \begin{equation}\label{factorisation0}
    \exp(tX) = g_+(t)^{-1} g_-(t), 
    \end{equation}
    with initial conditions $g_\pm(0) = e$, and $t$ in an interval around the origin.

Let us mention that the decomposition $\g = \g_+ \oplus \g_-$ implies that there exist neighborhoods of the unit element $e\in\mathcal{V}_G\subset G$, $e\in\mathcal{V}_{G_+}\subset G_+$, and $e\in\mathcal{V}_{G_-}\subset G_-$ such that the multiplication map $\textbf{m}: \mathcal{V}_{G_+}\times\mathcal{V}_{G_-}\rightarrow \mathcal{V}_G$ is a diffeomorphism.
Therefore, the factorization problem \eqref{factorisation0} admits a solution, at least locally.

We will need the following Lemma, analogous to \cite[Lemma~2.9]{adler}.

\begin{lemma}
   Let $H\in \Cp^\infty(\bg)$ be an $\Ad^*_G$-invariant function.  
   For any $\mu\in \bg$ and any $g\in G$, one has
    \begin{equation}\label{Hinv}
    D_{\Ad^*_g(\mu)} H = \Ad_g(D_\mu H).
    \end{equation}
\end{lemma}

\begin{proof}
An $\Ad^*_G$-invariant function $H$ on $\bg$ satisfies $H(\Ad^*_g(\mu)) = H(\mu)$ for any $g\in G$ and any $\mu \in \bg$. Therefore
   \[
   \begin{array}{ll}
    \ad^*_{D_\mu H}\mu (X) &= \langle \mu, [D_\mu H, X]\rangle = -\langle \ad^*_X \mu, D_\mu H\rangle = -D_\mu H\left(\ad^*_X \mu\right)\\ &  = \frac{d}{dt}_{|t = 0} H\left(\Ad^*_{\exp -tX}\mu\right) = \frac{d}{dt}_{|t = 0} H\left(\mu\right) = 0
    \end{array}
    \]
    for any $X\in \g$.
    Moreover, 
    by differentiating the identity $H\left(\Ad^*_{g}\mu\right) = H(\mu)$ 
    at $\mu$, one obtains
    \[
    D_{\Ad^*_{g}\mu} H\circ D_\mu Ad^*_{g} = D_\mu H.
    \]
    Since $\Ad^*_{g}: \bg\rightarrow \bg$ is linear, one has
    \[
    D_{\Ad^*_{g}\mu} H\circ Ad^*_{g} = D_\mu H.
    \]
    Consequently, for any $\eta\in\bg$,
    \[
    \langle D_\mu H, \eta\rangle = \langle D_{\Ad^*_{g}\mu} H, Ad^*_{g}\eta\rangle = \langle \Ad_g^{-1} D_{\Ad^*_{g}\mu} H, \eta\rangle.
    \]
    Therefore
    \[
    D_\mu H = \Ad_g^{-1} D_{\Ad^*_{g}\mu} H,
    \]
    which is equivalent to \eqref{Hinv}.
    
\end{proof}

\begin{theorem}\label{integrale_curve} 
Let $\bg$ be a Lie--Poisson space with respect to a Banach Lie algebra $\g$ which admits a decomposition  $\g = \g_+ \oplus \g_-$ into the sum of two Lie subalgebras, and 
consider the  $R$-matrix $R = p_+ - p_-$. Suppose that $\bg$ is also  a Banach Lie--Poisson space with respect to  $\g_R$, and that there exists a Banach Lie group $G$ with Lie algebra $\g$, and Banach Lie subgroups $G_+$ and $G_-$ with Lie algebras $\g_+$ and $\g_-$ respectively. Denote by $\mathcal{A}$ the algebra of smooth functions on $\bg$ with derivative in $\g$.

    Then, for  an $\Ad^*_G$-invariant function $H\in \mathcal{A}$, the integral curve of the Hamiltonian vector field $X_H = \{\cdot, H\}_R$, starting at $\mu_0\in\bg$, is given by
    \begin{equation}\label{solution}
    \mu(t) = \Ad^*_{g_+(t)}\mu_0 = \Ad^*_{g_-(t)}\mu_0,
    \end{equation}
    where $g_+(t)\in G_+$ and $g_-(t)\in G_-$ are the smooth curves solving the factorization problem
    \begin{equation}\label{factorisation}
    \exp(-t D_{\mu_0}H) = g_+(t)^{-1} g_-(t), \textrm{  with initial conditions } g_\pm(0) = e,
    \end{equation}
    and $t$ in an interval around the origin.
\end{theorem}
\begin{remark}
    Using Theorem~\ref{thmepsilonnon0}, one gets immediately the integral curves of the Hamiltonian vector fields $X_{H_\varepsilon}$ given in \eqref{XHepsilon} by replacing in equation~\eqref{solution} $\mu(t) $ by $x(t) + \varepsilon$ and $\mu_0$ by $x_0 + \varepsilon$.
\end{remark}

\begin{proof}${}$

    \begin{itemize}
        \item 
    Let us first show that \[\Ad^*_{g_+(t)}\mu_0 = \Ad^*_{g_-(t)}\mu_0.\]  Since $\Ad^*_{\exp(-t D_{\mu_0}H)} = \Ad^*_{g_+(t)^{-1}} \Ad^*_{g_-(t)}$, this will follow from the fact that 
    \begin{equation}\label{Admu}
    \Ad^*_{\exp(-t D_{\mu_0}H)}\mu_0 = \mu_0.
    \end{equation}
    To prove \eqref{Admu} recall that $H$ is $\Ad^*_G$-invariant, hence
    is preserved by the coadjoint action of $\g$. Consequently  by \eqref{hinv}$, \ad^*_{D_{\mu_0}H}\mu_0 = 0$. Then, for any $X\in\g$,
    \[
    \begin{array}{ll}
    \langle \Ad^*_{\exp(-t D_{\mu_0}H)}\mu_0, X\rangle & = \langle \mu_0, \Ad_{\exp(t D_{\mu_0}H)}X \rangle\\
    & =  \langle \mu_0, \exp(t \ad_{D_{\mu_0}H}X \rangle\\
    & =  \langle \mu_0, X\rangle + \left\langle \mu_0, \ad_{D_{\mu_0}H}\left(\sum_{n = 1}^{+\infty} \frac{(t \ad_{D_{\mu_0}H})^{n-1}}{n!}(X)\right)\right\rangle\\
    & = \langle \mu_0, X\rangle + \left\langle \ad^*_{D_{\mu_0}H}\mu_0, \left(\sum_{n = 1}^{+\infty} \frac{(t \ad_{D_{\mu_0}H})^{n-1}}{n!}(X)\right)\right\rangle\\ & = \langle \mu_0, X\rangle,
    \end{array}
    \]
    which implies that $\Ad^*_{\exp(-t D_{\mu_0}H)}\mu_0 = \mu_0$. 
    
     \item Let us prove that $\mu(t) = \Ad^*_{g_+(t)}\mu_0$ is an integral curve of $X_H(\mu) =  \ad^*_{(D_\mu H)_+}\mu$. First, one has
    \[
    \begin{array}{ll}
    \frac{d}{dt}_{| t= t_0}\Ad^*_{g_+(t)}\mu_0& = \frac{d}{dt}_{| t= t_0}\Ad^*_{g_+(t)g_+(t_0)^{-1}}\Ad^*_{g_+(t_0)}\mu_0\\
    &   = \ad^*_{\left(\frac{d}{dt}_{|t = t_0}{g_+}(t)\right)\cdot g_+(t_0)^{-1}}\mu(t_0),
    \end{array}
    \]
    where $\left(\frac{d}{dt}_{|t = t_0}{g_+}(t)\right)\cdot g_+(t_0)^{-1}$ denotes the differential of the right translation by $g_+(t_0)^{-1}$ applied to the vector $\frac{d}{dt}_{|t = t_0}{g_+}(t)\in T_{g_+(t_0)}G$.
 Comparing with the Hamiltonian vector field \eqref{hamvec}, we have to prove that 
 \[
 \left(\frac{d}{dt}_{|t = t_0}{g_+}(t)\right)\cdot g_+(t_0)^{-1} = (D_{\mu(t)} H)_+.
 \]
    Differentiating equation~\eqref{factorisation}
   leads to 
   \[
    (-D_{\mu_0}H)\cdot g_+(t_0)^{-1} g_-(t_0) = -g_+(t_0)^{-1}\cdot \left(\frac{d}{dt}_{|t = t_0}g_+(t)\right)\cdot  g_+(t_0)^{-1}g_-(t_0) + g_+(t_0)^{-1}  \cdot \left(\frac{d}{dt}_{|t = t_0}g_-(t)\right),
   \]
   which, after right multiplication by the inverse of $g_-(t_0)$, gives
   \[
    (-D_{\mu_0}H)\cdot g_+(t_0)^{-1} = -g_+(t_0)^{-1}\cdot \left(\frac{d}{dt}_{|t = t_0}g_+(t)\right)\cdot g_+(t_0)^{-1} + g_+(t_0)^{-1}\cdot \left(\frac{d}{dt}_{|t = t_0}g_-(t)\right)\cdot g_-(t_0)^{-1}.
   \]
   After left multiplication by $g_+(t_0)$, one obtains
   \[
    g_+(t_0)\cdot (-D_{\mu_0}H)\cdot g_+(t_0)^{-1} = - \left(\frac{d}{dt}_{|t = t_0}g_+(t)\right)\cdot g_+(t_0)^{-1} + \left(\frac{d}{dt}_{|t = t_0}g_-(t)\right)\cdot g_-(t_0)^{-1}.
   \]
   By equation~\eqref{Hinv} 
    together with $\mu(t) = \Ad^*_{g_+(t)}\mu_0$, one has
    \[
    D_{\mu(t_0)} H = D_{\Ad^*_{g_+(t_0)}\mu_0} H = \Ad_{g_+(t_0)}(D_{\mu_0} H) =  g_+(t_0)\cdot (D_{\mu_0}H)\cdot g_+(t_0)^{-1} ,
    \]
    hence 
    \[
    -D_{\mu(t)} H = - \left(\frac{d}{dt}_{|t = t_0}g_+(t)\right)\cdot g_+(t_0)^{-1} + \left(\frac{d}{dt}_{|t = t_0}g_-(t)\right)\cdot g_-(t_0)^{-1}.
    \]
    Taking the projection on $\g_+$ of previous equality gives the result.
    \end{itemize}
\end{proof}

\section{Rota-Baxter Banach Lie algebras and Rota-Baxter Banach Lie groups}\label{sec:B}

Rota-Baxter Lie algebras and the corresponding Rota-Baxter Lie groups were studied extensively in 
\cite{guo2021}. Let us recall some of the results connected to the factorization problem~\ref{factpb} and Theorem~\ref{integrale_curve}.

\subsection{Rota-Baxter Banach Lie algebras}

\begin{definition}
A Rota-Baxter operator of weight $\lambda$ on a Banach Lie algebra $(\g, [\cdot, \cdot]_\g)$ is a linear operator $B:\g\rightarrow \g$ such that the following identity holds
 \begin{equation}\label{baxter}
    [Bx, By]_\g = B[Bx, y]_\g + B[x, By]_\g + \lambda B[x, y]_\g,
\end{equation}
    for all $x, y \in \g$. 
 A Rota-Baxter Banach Lie algebra of weight $\lambda$ is a Banach Lie algebra $(\g, [\cdot, \cdot]_\g)$ endowed with a 
 Rota-Baxter operator $B$ of weight $\lambda$.
\end{definition}

\begin{example}
    Suppose that $\g$ is the sum of two closed subalgebras: $\g = \g_+\oplus \g_-$ and denote by $p_\pm$ the projections on each factor. Than $B = -p_\pm$ is a Rota-Baxter operator of weight $1$ and $B = p_\pm$ is a Rota-Baxter operator of weight $-1$. 
\end{example}

\begin{proposition}
For a Rota-Baxter Banach Lie algebra $(\g, [\cdot, \cdot]_\g, B)$ of weight $1$, the following bracket
\begin{equation}\label{bracketB}
    [x, y]_B = [Bx, y]_\g + [x, By]_\g + [x,y]_\g
\end{equation}
    is a Lie bracket on $\g$, called the Baxter bracket associated to $B$.
\end{proposition}

\begin{proof}
    One has
    \begin{align*}
    [[x, y]_B, z]_B & = [[Bx, y]_\g + [x, By]_\g + [x,y]_\g, z]_B\\
    & = [B \left([Bx, y]_\g + [x, By]_\g + [x,y]_\g\right), z]_\g\\ & \quad+ [[Bx, y]_\g + [x, By]_\g + [x,y]_\g, Bz]_\g + \\&\quad[[Bx, y]_\g + [x, By]_\g + [x,y]_\g, z]_\g
    \\
    & = [[Bx, By]_\g, z]_\g  + [[Bx, y]_\g, Bz]_\g + [[x, By]_\g, Bz]_\g + [[x,y]_\g, Bz]_\g  \\&\quad+[[Bx, y]_\g, z]_\g + [[x, By]_\g, z]_\g + [[x,y]_\g, z]_\g
    \end{align*}
    and similarly
    \begin{align*}
    [[y, z]_B, x]_B & = [[By, Bz]_\g, x]_\g  + [[By, z]_\g, Bx]_\g + [[y, Bz]_\g, Bx]_\g + [[y,z]_\g, Bx]_\g  \\&\quad+[[By, z]_\g, x]_\g + [[y, Bz]_\g, x]_\g + [[y,z]_\g, x]_\g
      \end{align*}
      \begin{align*}
    [[z, x]_B, y]_B & = [[Bz, Bx]_\g, y]_\g  + [[Bz, x]_\g, By]_\g + [[z, Bx]_\g, By]_\g + [[z,x]_\g, By]_\g  \\&\quad+[[Bz, x]_\g, y]_\g + [[z, Bx]_\g, y]_\g + [[z,x]_\g, y]_\g.
      \end{align*}
      The Jacobi identity for $[\cdot, \cdot]_B$ then follows from the following Jacobi identities for $[\cdot,\cdot]_\g$:
      \begin{align*}
          [[Bx, By]_\g, z]_\g  + [[By, z]_\g, Bx]_\g + [[z, Bx]_\g, By]_\g = 0\\
           [[Bx, y]_\g, Bz]_\g + [[y, Bz]_\g, Bx]_\g + [[Bz, Bx]_\g, y]_\g = 0\\
          [[By, Bz]_\g, x]_\g  + [[Bz, x]_\g, By]_\g + [[x, By]_\g, Bz]_\g  = 0
      \end{align*}
      and 
      \begin{align*}
           [[x,y]_\g, Bz]_\g + [[y, Bz]_\g, x]_\g + [[Bz, x]_\g, y]_\g  = 0\\
           [[Bx, y]_\g, z]_\g + [[y,z]_\g, Bx]_\g + [[z, Bx]_\g, y]_\g  = 0\\
           [[x, By]_\g, z]_\g + [[By, z]_\g, x]_\g + [[z,x]_\g, By]_\g = 0\\
           [[x,y]_\g, z]_\g + [[y,z]_\g, x]_\g+ [[z,x]_\g, y]_\g = 0.         
      \end{align*}
\end{proof}

\begin{proposition}
    An operator $B:\g\rightarrow \g$ on a Lie algebra $\g$ is a Rota-Baxter operator of weight $1$ if and only if the operator $R = \id + 2B$ satisfies the modified classical Yang--Baxter equation~\ref{mCYBE}. Moreover, the corresponding Lie brackets on $\g$ are equal: $[\cdot, \cdot]_B = [\cdot, \cdot]_R$.
\end{proposition}

\begin{proof}
 Suppose that $[Bx, By]_\g = B[Bx, y]_\g + B[x, By]_\g +  B[x, y]_\g$ and set $R = \id + 2B$. The LHS of the modified classical Yang--Baxter equation for $R$ reads:
\begin{align*}
[Rx, Ry]_\g & = [x + 2Bx, y + 2By]_g = [x, y]_\g + 2[Bx, y]_\g + 2[x, By]_\g + 4[Bx, By]_\g\\
&  \quad = [x, y]_\g + 2[Bx, y]_\g + 2[x, By]_\g+ 4B[Bx, y]_\g + 4B[x, By]_\g + 4B[x, y]_\g,
\end{align*}
  The RHS of the modified classical Yang--Baxter equation for $R$ reads:
\begin{align*}
R\left([Rx, y]_\g + [x, Ry]_\g\right) - [x, y]_\g  = [Rx, y]_\g + [x, Ry]_\g + 2B\left([Rx, y]_\g + [x, Ry]_\g\right) - [x, y]_\g\\
=  2[Bx, y]_\g  + 2[x, By]_\g + 2B[x, y]_\g + 4B[Bx, y]_g +2B[x, y]_\g + 4B[x, By]_\g + [x,y]_\g.\\
\end{align*}
The equivalence is then easily checked.
\end{proof}

\begin{example}
    For a Banach Lie algebra $\g$ which is the sum of two Banach Lie subalgebras $\g = \g_+\oplus \g_-$, the Lie bracket associated to $B = -p_-$ is 
    \begin{equation}
        [x, y]_{B = -p_-} = [x_+, y_+]_{\g_+} - [x_-, y_-]_{\g_-}.
    \end{equation}
    Hence it coincides with the $R$-bracket defined in equation~\eqref{R-bracket}.
\end{example}

\subsection{Rota-Baxter Lie groups}

\begin{definition}
    A Rota-Baxter Banach Lie group is a Banach Lie group $G$ endowed with a smooth map $\mathfrak{B}: G  \rightarrow G$ satisfying
    \begin{equation}
        \mathfrak{B}(g_1)\mathfrak{B}(g_2) = \mathfrak{B}\left(g_1 \textrm{Ad}_{\mathfrak{B}(g_1)}g_2\right),
    \end{equation}
    for all $g_1, g_2 \in G$.
\end{definition}

The following Lemma is the Banach Lie version of Lemma~2.6 in \cite{guo2021} and is straightforward.
\begin{lemma}\label{rota}
Let $G$ be a Banach Lie group and $G_+$ and $G_-$ two Banach subgroups such that $G = G_+ G_-$ and $G_+\cap G_- = \{e\}$. Define $\mathfrak{B}: G \rightarrow G$ by 
\[
\mathfrak{B}(g) = g_-^{-1}, \forall g = g_+ g_-, \textrm{  where  } g_+\in G_+, g_-\in G_-.
\]
Then $(G, \mathfrak{B})$ is a Rota-Baxter Banach Lie group.
\end{lemma}
The link between Rota-Baxter Banach Lie groups and  Rota-Baxter Banach Lie algebra is given by the following proposition. We refer the reader to Theorem~2.9 in \cite{guo2021} for the proof which extends to the Banach setting without difficulty.
\begin{theorem}
  Given a Rota-Baxter Banach Lie group $(G, \mathfrak{B})$ with Lie algebra $\g$. Denote by $B = \mathfrak{B}_{*e}: \g\rightarrow \g$ the tangent map of $\mathfrak{B}$ at the unit element $e$. Then $(\g, B)$ is a Rota-Baxter Lie algebra of weight $1$.
\end{theorem}
The following Proposition is a straightforward generalization of Proposition~2.13(i) in \cite{guo2021} to the Banach setting.
\begin{proposition}
    Let $(G, \mathfrak{B})$ be a Rota-Baxter Banach Lie group.  Endow $G$ with the multiplication
    \begin{equation}\label{*}
    g_1*g_2 = g_1\Ad_{\mathfrak{B}(g_1)}g_2, \quad \forall g_1, g_2 \in G.
    \end{equation}
    Then $(G, *)$ is also a Banach Lie group. Its Lie algebra is $(\g, [\cdot, \cdot]_B)$, where $B = \mathfrak{B}_{*e}$, and $[\cdot, \cdot]_B$ is given by \eqref{bracketB}. 
\end{proposition}
The following Proposition is the Banach version of Corollary~2.14 in \cite{guo2021}: 
\begin{proposition}
    In the setting of Lemma~\ref{rota}, the Lie group $G = G_+G_-$ can be endowed with a new Lie group structure with group multiplication $*: G\times G \rightarrow G$ given by 
    \begin{equation}
        g * h = (g_+g_-) * (h_+h_-) = g_+g_- g_-^{-1} h_+h_- g_- =  g_+ h_+ h_- g_-, \forall g, h \in G.
    \end{equation}
    The corresponding Lie bracket on the Lie algebra $\g$ is given by
    \[
    [x, y]_{B = -p_{-}} = [x_+, y_+]_{\g_+} - [x_-, y_-]_{\g_-}.
    \]
\end{proposition}
\begin{proof}
    Apply equation~\eqref{*} to $\mathfrak{B}(g) = g_-^{-1}$, where $g = g_+g_-\in G$.
\end{proof}

\begin{example} Let us recall from \cite[Example 2.8]{guo2021} the following examples.  We will see in Section~\ref{Iwasawa_infinite} a Banach version of these examples. 
    \begin{enumerate}
        \item Lie group $SL(n, \C)$ of complex $n\times n$ matrices with determinant~$1$ factorizes as
        \[
        SL(n, \C) = SU(n) SB(n, \C),
        \]
        where $SU(n)$ is the real Lie group of unitary matrices with determinant $1$ and $SB(n, \C)$ is the real Lie group of all upper triangular matrices $SL(n, \C)$ with positive coefficients on the diagonal.
        Then $SL(n, \C)$ is a Rota-Baxter Lie group for $\mathfrak{B}(ub) = b^{-1}$, where $u \in SU(n)$ and $b \in SB(n, \C)$.
        \item More generally, using the Iwasawa decomposition $G = KAN$ of a semi-simple group $G$, one obtains a Rota-Baxter Lie group $(G, \mathfrak{B})$ where the map $\mathfrak{B}: G \rightarrow G$ is defined by $\mathfrak{B}(kan) = (an)^{-1}$, for $k\in K, a\in A, n\in N$.
        \end{enumerate}     
\end{example}
\section{Nijenhuis operators on Banach Lie algebras}\label{sec:N}

\subsection{Linear Nijenhuis operators and associated Lie brackets}
Nijenhuis operators were applied in the theory of integrable, see e.g. \cite{magri90,bolsinov-n1,kosmann-bial}. In the Banach setting they were studied in \cite{GLT-NN} and used in \cite{GLT-NN-proc,grabowski25}. In this section, we recall the Lie bracket associated to a Nijenhuis operator and the relation to Rota-Baxter algebras. We follow the presentation given in \cite{magri90} for the finite-dimensional case. The results of this section will be applied to the semi-infinite Toda lattice in section~\ref{sec:Toda}.

\begin{definition}
    A linear operator $N:\g \rightarrow \g$ on a Banach Lie algebra $\g$ is called a linear Nijenhuis operator on $\g$ if 
    \begin{equation}\label{N}
        [Nx, Ny]_\g = N[Nx, y]_\g + N[x, Ny]_\g - N^2[x, y]_\g, \quad \forall x, y \in \g.
    \end{equation}
\end{definition}

More generally one has the following definition:
\begin{definition}\label{torJ} Let $\mathcal M$ be any smooth Banach manifold and let $\Nc:T\mathcal M\to T\mathcal M$ be a smooth Banach vector bundle map. The \emph{Nijenhuis torsion} of $\Nc$ is defined as 
\[
\Omega_{\Nc}(X,Y)=\Nc [\Nc X,Y]+\Nc [X,\Nc Y]-[\Nc X,\Nc Y]-\Nc^{\; 2}[X,Y]
\]
for $X,Y$ vector fields in $\mathcal M$ and where $[\cdot, \cdot]$ denotes the bracket of vector fields. We say that $\Nc$ is a \emph{Nijenhuis operator} on $\mathcal M$ if its torsion vanishes.
\end{definition}

\begin{proposition}
Consider a linear Nijenhuis operator on the Banach Lie algebra $\g$ of a Banach Lie group $G$, and define a Banach vector bundle map $\Nc: TG \rightarrow TG$ on the tangent bundle $TG$ by 
    \begin{equation}\label{Np}
\Nc_g = (L_g)_* N \, (L_g)_*^{-1},
\end{equation}
where $L_g$ denotes the left translation by $g\in G$. Then $\Nc$ is a Nijenhuis operator on $G$.
\end{proposition}

\begin{proof}
    This is a direct consequence of Theorem~3.6 in \cite{GLT-NN} with $K = \{e\}$. 
\end{proof}

\subsection{Compatibility between the usual bracket and the \texorpdfstring{$N$}{N}-bracket}
The following bracket related to a Nijenhuis operator was introduced in \cite{magri90}.

\begin{proposition}
Given a linear Nijenhuis operator $N$ on a Banach Lie algebra $\g$, one can define a new Lie bracket on $\g$ by
\begin{equation}\label{def:N}
    [x, y]_N = [Nx, y]_\g + [x, Ny]_\g - N [x, y]_\g, 
\end{equation}
where $x, y \in \g$.
\end{proposition}
\begin{proof}
    One has
    \begin{align*}
[[x, y]_N, z]_N & = [[Nx, y]_\g + [x, Ny]_\g - N [x, y]_\g, z]_N \\
& = [N\left([Nx, y]_\g + [x, Ny]_\g - N [x, y]_\g\right), z]_\g \\ & \quad+ [[Nx, y]_\g + [x, Ny]_\g - N [x, y]_\g, Nz]_\g \\ & \quad- N[[Nx, y]_\g + [x, Ny]_\g - N [x, y]_\g, z]_\g\\
& = [[Nx, Ny]_\g, z]_\g + [[Nx, y]_\g, Nz]_\g + [[x, Ny]_\g, Nz]_\g - [N [x, y]_\g, Nz]_\g \\ & \quad- N[[Nx, y]_\g, z]_\g - N[[x, Ny]_\g, z]_\g + N[N [x, y]_\g, z]_\g\\
& = [[Nx, Ny]_\g, z]_\g + [[Nx, y]_\g, Nz]_\g + [[x, Ny]_\g, Nz]_\g \\ & \quad- N[[Nx, y]_\g, z]_\g - N[[x, Ny]_\g, z]_\g -N[[x, y]_\g, Ny]_\g + N^2[[x, y]_\g, z]_\g\\
    \end{align*}
    where we have used equation~\eqref{N} twice. 
    Similarly
    \begin{align*}
[[y, z]_N, x]_N & = [[Ny, Nz]_\g, x]_\g + [[Ny, z]_\g, Nx]_\g + [[y, Nz]_\g, Nx]_\g \\ & \quad- N[[Ny, z]_\g, x]_\g - N[[y, Nz]_\g, x]_\g -N[[y, z]_\g, Nx]_\g + N^2[[y, z]_\g, x]_\g\\
    \end{align*}
    and
    \begin{align*}
[[z, x]_N, y]_N & = [[Nz, Nx]_\g, y]_\g + [[Nz, x]_\g, Ny]_\g + [[z, Nx]_\g, Ny]_\g \\ & \quad- N[[Nz, x]_\g, y]_\g - N[[z, Nx]_\g, y]_\g -N[[z, x]_\g, Ny]_\g + N^2[[z, x]_\g, y]_\g.\\
    \end{align*}
    The Jacobi identity for $[\cdot, \cdot]_N$ then follows from the linearity of $N$ and the following Jacobi identities for $[\cdot,\cdot]_\g$:
    \begin{align*}
        [[Nx, Ny]_\g, z]_\g + [[Ny, z]_\g, Nx]_\g + [[z, Nx]_\g, Ny]_\g = 0\\
        [[Nx, y]_\g, Nz]_\g +  [[y, Nz]_\g, Nx]_\g + [[Nz, Nx]_\g, y]_\g = 0\\
         [[x, Ny]_\g, Nz]_\g + [[Ny, Nz]_\g, x]_\g + [[Nz, x]_\g, Ny]_\g  = 0\\
    \end{align*}
    as well as
     \begin{align*}
[[Nx, y]_\g, z]_\g +[[y, z]_\g, Nx]_\g + [[z, Nx]_\g, y]_\g  = 0\\
[[x, Ny]_\g, z]_\g +[[Ny, z]_\g, x]_\g+[[z, x]_\g, Ny]_\g= 0\\
N[[x, y]_\g, Ny]_\g+N[[y, Nz]_\g, x]_\g +N[[Nz, x]_\g, y]_\g  = 0\\
[x, y]_\g, z]_\g + [y, z]_\g, x]_\g +[z, x]_\g, y]_\g  =0.\\
     \end{align*}
\end{proof}
\begin{definition}
    Two Lie brackets $[\cdot, \cdot]_1$ and $[\cdot, \cdot]_2$ on the same Banach space $\g$ are said to be compatible if their sum is a Lie bracket on $\g$.
\end{definition}
\begin{remark}
    If two Lie brackets $[\cdot, \cdot]_1$ and $[\cdot, \cdot]_2$ are compatible Lie brackets on a Banach space $\g$ then for any $\lambda$ in $\R$ or $\C$, 
    $[\cdot, \cdot]_\lambda = [\cdot, \cdot]_1 + \lambda [\cdot, \cdot]_2$ is a Lie bracket.
\end{remark}

\begin{proposition}\label{compatibilityGN}
    For any linear Nijenhuis operator $N$ on a Banach Lie algebra $\g$, the Lie bracket $[\cdot, \cdot]$ and $[\cdot, \cdot]_N$ are compatible.
\end{proposition}

\begin{proof}
    Denote by $[\![\cdot, \cdot]\!] = [\cdot, \cdot] + [\cdot, \cdot]_N$. One has
    \[
    \begin{array}{ll}
[\![x, [\![y, z]\!] ]\!] &= [x, [\![y, z]\!] ]  + [x, [\![y, z]\!] ]_N\\
& = [x, [y, z] ]  + [x, [y, z]_N ]  + [x, [y, z] ]_N + [x, [y, z]_N ]_N\\
\end{array}
    \]
    The sum of the first and last terms over cyclic permutations of $x, y, z$ vanish by the Jacobi identity for $[\cdot, \cdot]$ and $[\cdot, \cdot]_N$. The middle terms can be written as
     \[
    \begin{array}{l}
    [x, [y, z]_N ]  + [x, [y, z] ]_N = \\
    = [x, [Ny, z] ] + [x, [y, Nz]] - [x, N[y,z] ]  + [Nx, [y, z] ] + [x, N[y,z]] - N[x, [y,z]]\\
   = [x, [Ny, z] ] + [x, [y, Nz]]   + [Nx, [y, z] ] - N[x, [y,z]].
    \end{array}
    \]
    The Jacobi identity for $[\![\cdot, \cdot]\!]$ then follows from the Jacobi identity for $[\cdot, \cdot]$.
\end{proof}

\begin{proposition}\label{GNcompatible}
Let $\bg$ be a Banach Lie--Poisson space with respect to a Banach Lie algebra $\g$, and let $N$ be a linear Nijenhuis operator on $\g$. If the dual map $N^*:\g^*\to\g^*$ preserves $\bg$
\[ N^*\bg\subset \bg,\] 
then $\bg$ is also a Banach Lie--Poisson space with respect to the Banach Lie algebra $(\g, [\cdot, \cdot]_N)$. Moreover the Lie--Poisson brackets on $\bg$ associated with $[\cdot, \cdot]$ and  $[\cdot, \cdot]_N$ are compatible.
\end{proposition}

\begin{proof}
  By Definition~\eqref{def:N}, the coadjoint representation with respect to $[\cdot, \cdot]_N$ reads
    \begin{equation}\label{coad_N}
    (\ad_N^*)_x = \big(\ad^*_{Nx} + N^*\ad^*_x - \ad^*_x N^*\big),
    \end{equation}
    where $x\in\g$. Since we assumed that $\bg$ is a Banach Lie--Poisson space with respect to $\g$, both $\ad^*_{Nx}$ and $\ad^*_x$ preserve $\bg$. Thus a sufficient condition to get a Banach Lie--Poisson structure on $\bg$ with respect to $\left(\g, [\cdot, \cdot]_N\right)$ is for $N^*$ to preserve $\bg$ as well. Moreover, by Proposition~\ref{GNcompatible}, since the sum of $[\cdot, \cdot]$ and to $[\cdot, \cdot]_N$ is a Lie bracket, the sum of the corresponding Lie--Poisson brackets on the space of smooth functions on $\bg$ with differential in $\g$ is the Lie--Poisson bracket associated with $[\cdot, \cdot]+[\cdot, \cdot]_N$.
\end{proof}
\begin{example}
    Suppose that $\g$ is a Banach Lie algebra with a decomposition $\g = \g_+ \oplus \g_-$ into the sum of two Banach Lie subalgebras, and denote by $p_\pm: \g \rightarrow \g_\pm$ the projections onto each factor. Then
    \begin{enumerate}
        \item $N = p_+ - p_-$ is a linear Nijenhuis operator on $\g$ with corresponding bracket
        \begin{equation}
            [x, y]_{N = p_+ - p_-} = [x_+, y_+] - [x_-, y_-] - (p_+-p_-)\left([x_+, y_-] + [x_-, y_+]\right),
        \end{equation}
        where $x = x_+ + x_-$, $y = y_+ + y_-$, $x_+, y_+\in \g_+$, $x_-, y_-\in \g_-$.
        \item $N = p_+$ is a linear Nijenhuis operator on $\g$ with corresponding bracket
        \begin{equation}
            [x, y]_{N = p_+} = [x_+, y_+] + p_-\left([x_+, y_-]+[x_-, y_+]\right).
        \end{equation}
        \item Similarly, $N= p_-$ is a linear Nijenhuis operator on $\g$ with corresponding bracket
        \begin{equation}
            [x, y]_{N = p_-} = [x_-, y_-] + p_+\left([x_+, y_-]+[x_-, y_+]\right).
        \end{equation}
    \end{enumerate}
\end{example}

\subsection{Idempotent Nijenhuis operators and Rota-Baxter operators}

\begin{proposition}
    An idempotent linear Nijenhuis operator $N = N^2$ on a Banach Lie algebra $\g$ is a Rota-Baxter operator of weight $-1$. 
\end{proposition}
\begin{proof}
    An linear Nijenhuis operator $N:\g\rightarrow\g$ on a Banach Lie algebra $\g$ satisfies
    \[
      [Nx, Ny]_\g = N[Nx, y]_\g + N[x, Ny]_\g - N^2[x, y]_\g, \quad \forall x, y \in \g.
    \]
    When $N$ is idempotent, $N^2 = N$, the previous identity reduces to equation~\eqref{baxter} with $\lambda = -1$.
\end{proof}
\begin{corollary}
    Consider an idempotent linear Nijenhuis operator $N = N^2$ on a Banach Lie algebra $\g$. Then $\g$ admits three Lie brackets:
    \begin{enumerate}
        \item the original Lie bracket $[\cdot, \cdot]_\g$;
        \item the Nijenhuis bracket (which is compatible with $[\cdot,\cdot]_\g$)
        \begin{equation}
             [x, y]_N = [Nx, y]_\g + [x, Ny]_\g - N [x, y]_\g, x, y \in \g;
         \end{equation}
        \item the Baxter bracket associated to $B = -N$
        \begin{equation}
         [x, y]_B = -[Nx, y]_\g - [x, Ny]_\g + [x,y]_\g, x, y \in \g.
        \end{equation}
    \end{enumerate}
\end{corollary}

\begin{remark}
    For the Nijenhuis operator $N = p_+$ corresponding to a decomposition $\g =\g_+ \oplus \g_-$ into the sum of two Banach Lie subalgebras, one has
   \begin{align}
    [x, y]_\g  &= [x_+, y_+]_\g  + [x_+, y_-]_\g  + [x_-, y_+]_\g  + [x_-, y_-]_\g  \\
    [x,  y]_{N = p_+} &= [x_+, y_+]_\g  +  p_-\left([x_+, y_-]_\g  + [x_-, y_+]_\g \right) \\
    [x, y]_{B = -p_+} &=   [x_-, y_-]_\g  - [x_+, y_+]_\g ,
    \end{align}
    where $x = x_+ + x_-$, $y = y_+ + y_-$, $x_+, y_+\in \g_+$, $x_-, y_-\in \g_-$.
    In particular, the restriction of all three brackets $[\cdot, \cdot]_\g $, $[x, y]_N$ and $[x, y]_{B = -p_+}$ to the subalgebra $\g_+$ are equal to the Lie bracket of $\g_+$. Moreover by Proposition~\ref{GNcompatible}, $[\cdot,\cdot]_\g $ and $[\cdot, \cdot]_{N = p_+}$ are compatible.
\end{remark}

\section{Lax equations associated with Banach--Poisson Lie groups}\label{sec:Lax}

\subsection{Lax equations are equations on adjoint orbits.}
Given a Banach Lie group $G$ with Banach Lie algebra $\g$, the adjoint orbit of an element $L_0\in \g$ is defined as
\[
\mathcal{O}_{L_0} = \{\Ad_g(L_0), g \in G\}.
\]
A tangent vector at $L\in \mathcal{O}_{L_0}$ is the differential of a smooth curve $L(t)$ in  $\mathcal{O}_{L_0}$ of the form
\[
L(t) = \Ad_{g(t)}L,
\]
where $g(t)$ is a smooth curve in $G$ with $g(0) = e$. Since
\[
\frac{d}{dt}_{|t = 0} L(t) =  \frac{d}{dt}_{|t = 0} \Ad_{g(t)}L = \left[\frac{d}{dt}_{t = 0}g(t), L\right],
\]
it follows that a tangent vector at $L\in \mathcal{O}_{L_0}$ is of the form
\[
[M, L] = \ad_M L,
\]
where $M\in \g$. 
An integral curve of a (possibly time-dependent) vector field tangent to an adjoint orbit is therefore what is called a Lax equation:
\[
\frac{d}{dt}L(t) = [M(t), L(t)].
\]

\subsection{From coadjoint action to adjoint action.}\label{ssec:ad}

Suppose that the Banach Lie algebra $\g$ of a Banach Lie group $G$ admits an $\Ad_G$-invariant non-degenerate continuous bilinear form
$\langle\cdot, \cdot\rangle: \g \times \g\rightarrow \C$. The non-degeneracy condition implies that the map $\iota$ defined as
\[
\begin{array}{llll}
\iota: & \g & \hookrightarrow & \g^*\\
& X & \mapsto & \langle X, \cdot\rangle
\end{array}
\]
is injective, hence $\g$ injects into its continuous dual $\g^*$. 
The $\Ad_G$-invariance means that for all $g\in G$ and $ X, Y\in \g$,
\begin{equation}\label{invariance1}
\langle\Ad_g X, \Ad_g Y \rangle = \langle X, Y \rangle. 
\end{equation}
After differentiation, one obtains that for all $X, Y, Z \in \g$, 
\begin{equation}\label{invariance2}
\langle [X, Y], Z\rangle + \langle Y, [X, Z]\rangle = 0.
\end{equation}
For $L\in \g$, consider the coadjoint orbit $\tilde{\mathcal{O}}_\mu$ of $\mu := \iota(L) = \langle L, \cdot\rangle$. A tangent vector to the coadjoint orbit $\tilde{\mathcal{O}}_\mu$ at $\mu$ is of the form
\[
\frac{d}{dt}_{|t = 0}\Ad^*_{g(t)}\mu = \ad^*_M\mu,
\]
where $g(t)$ is any smooth curve in $G$ with $g(0) = e$ and $\frac{d}{dt}_{|t = 0}g(t) =  M\in \g$.
Note that the covector $\ad^*_M\mu$ acts on $Y\in \g$ by
\[
\ad^*_M\mu(Y) = \mu(\ad_M Y) = \mu([M, Y]) = \langle L, [M, Y]\rangle.
\]
By equation \eqref{invariance2}, we get
\[
\ad^*_M\mu(Y) = \langle L, [M, Y]\rangle = -\langle [M, L], Y\rangle,
\]
which can be also written as
\[
\ad^*_M\mu = \ad^*_M\langle L, \cdot\rangle = \langle -[M, L], \cdot\rangle.
\]
In other words
\begin{equation}\label{equ:iota}
\ad^*_M\iota(L) = \iota(-\ad_M L).
\end{equation}
In particular $\iota(\g)$ is stable by the coadjoint action of $\g$. 
Moreover, for $\mu = \iota(L) = \langle L, \cdot\rangle$,  by \eqref{invariance1},
\[
\Ad^*_g\mu (Y) = \langle L, \Ad_g Y\rangle = \langle \Ad_{g^{-1}} L,  Y\rangle,
\]
hence 
\[
\Ad^*_g\iota(L) = \iota(\Ad_{g^{-1}} L).
\]
It follows that the coadjoint orbit of $\mu = \iota(L) = \langle L, \cdot\rangle$ is the image by $\iota$ of the adjoint orbit of $L$:
\[
\tilde{\mathcal O}_{\iota(L)} = \iota\left(\mathcal{O}_L\right).
\]
In conclusion, in the presence of an $\Ad_G$-invariant non-degenerate continuous pairing on $\g$, equations on coadjoint orbits
\[
\frac{d}{dt}\mu = \ad^*_M\mu, M\in\g
\]
can be reformulated in Lax form when $\mu = \iota(L) \in \iota(\g)$
\[
\frac{d}{dt} L = -\ad_M L = [L, M].
\]

\subsection{Lax equations associated with \texorpdfstring{$R$}{R}-matrices}

Let us consider a Banach Lie algebra $\g$ admitting a non-degenerate continuous bilinear form
$\langle\cdot, \cdot\rangle: \g \times \g\rightarrow \g$ satisfying the invariance by adjoint action of $\g$ given in  \eqref{invariance2}. Denote by $\iota$ the injective map
\begin{equation}\label{iota}
\begin{array}{llll}
\iota: & \g & \hookrightarrow & \g^*\\
& X & \mapsto & \langle X, \cdot\rangle.
\end{array}
\end{equation}
By Corollary~\ref{g_lie_poisson}, $\g$ admits a Banach Lie--Poisson bracket $\pb:\mathcal{A}\times\mathcal{A}\rightarrow \mathcal{A}$ on smooth functions on $\g$ with differential in $\iota(\g)$
 \begin{equation}\label{Adef}
 \mathcal{A} := \{F \in \Cp^{\infty}(\g)~|~D_xF\in\iota(\g)\subset\g^*, \forall x\in\g\}
 \end{equation}
defined by
\begin{equation}\label{pb}
\{ F, H \}(x) = \langle x, [\nabla_x F, \nabla_x H]_\g \rangle,  \forall F,G\in\mathcal{A},
\end{equation}
where $\nabla_xF\in\g$ is defined by $\iota(\nabla_x F) = D_xF\in\iota(\g)$.
Let us translate the content of Section~\ref{sec:Involutivity} in this particular case.

\begin{theorem}\label{thm:aks_on_G}
Consider a Banach Lie  algebra $\g$ with an $\ad_\g$-invariant non-degenerate continuous bilinear map
$\langle\cdot, \cdot\rangle: \g \times \g\rightarrow \g$, and a $R$-matrix $R$ on $\g$.
    Suppose 
    the dual map $R^*:\g^*\to\g^*$ preserves $\iota(\g)\subset\g^*$ where $\iota$ is defined by \eqref{iota}.
    Then $\g$ admits a Banach Lie--Poisson bracket 
\begin{equation}\label{pb-R}
\{ F, G \}_R(x) = \langle x, [\nabla_xF, \nabla_x G]_R \rangle
\end{equation}
defined on functions $F, G\in\mathcal{A}$ \eqref{Adef}, i.e. with differential in $\iota(\g)\subset\g^*$.
    Consider $\ad_\g$-invariant functions $F,G\in \mathcal{A}$.
    Then we have:
    \begin{enumerate}
    \item[(1)] $\{F,G\}_R = 0$.
    \item[(2)] The flow of the Hamiltonian vector fields associated with $F\in \mathcal{A}$ with respect to $\pb_R$ is the solution of the following Lax equation
    \[
    \frac{d x}{dt}  = X_F(x) =  \frac{1}{2}\left[x, R \nabla_x F\right]_\g
    \]
    \item[(2')] If $\g = \g_+\oplus \g_-$ as sum of Banach Lie algebras and  $R = p_+ - p_-$, then 
    \[
    X_F(x) =   \pm \left[x, \left(\nabla_x F\right)_\pm\right]_\g.
    \]
    \item[(3)] Suppose that $G$ is a Banach Lie group with Lie algebra $\g$ which can be decomposed as the product of two Banach Lie subgroups $G_+$ and $G_-$, $G = G_+G_-$ with Lie algebras $\g_+$ and $\g_-$. Then, for  an $\Ad_G$-invariant function $H\in \mathcal{A}$, the integral curve of the Hamiltonian vector field $X_H = \{\cdot, H\}_R$, starting at $x_0$, is given by
    \[
    x(t) = \Ad_{g_+(t)}x_0 = \Ad_{g_-(t)}x_0,
    \]
    where $g_+(t)\in G_+$ and $g_-(t)\in G_-$ are the smooth curves solving the factorization problem
    \begin{equation}\label{factorisation-R}
    \exp(-t \nabla_{x_0}H) = g_+(t)^{-1} g_-(t), \textrm{  with initial conditions } g_\pm(0) = e.
    \end{equation}
\end{enumerate}
\end{theorem}

In the next section, we apply previous theorem to a particular group decomposition, known as Iwasawa decomposition.

\subsection{Lax equations associated with Iwasawa Banach Poisson--Lie groups}\label{Iwasawa_infinite}

\subsubsection{Iwasawa Decomposition for $GL^p(\mathcal H)$}
The existence of Iwasawa decompositions for infinite-dimensional Lie groups consisting of bounded operators on a separable Hilbert space is not guaranteed and is the topic of active research.
In the present paper, we are interested in the groups $GL^p(\mathcal H)$ where $1<p<+\infty$. Endow the separable Hilbert space $\mathcal{H}$ with an orthonormal basis $\{|n\rangle\}_{n=1}^{\infty}$. The following result is a direct consequence of Theorem~4.5 in \cite{beltita-iwasawa} (see also Example A.4. in \cite{beltita-iwasawa})
together with the fact that Schatten Ideals $L^p(\mathcal{H})$ have a non-trivial Boyd index \cite[Section~2]{beltita2010} for $1<p<+\infty$.

\begin{theorem}[{\cite[Theorem~4.5]{beltita-iwasawa}}]\label{thm:Iwasawa}
    For $1<p<+\infty$, consider the Banach Lie group  \[GL^p(\mathcal H) = (\Id + L^p(\mathcal H))\cap GL(\mathcal H)\]
    and its subgroups
    \[
    \begin{array}{l}
    U_p(\mathcal{H}) = \{u \in GL^p(\mathcal H)~|~ u^* = u^{-1}\}\\
    A_p(\mathcal{H}) = \{ a \in GL^p(\mathcal H)~|~ a|n\rangle \in \mathbb{R}^{+*}|n\rangle\}\\
    N_p(\mathcal{H}) = \{ g \in GL^p(\mathcal H)~|~  g|n\rangle \in |n\rangle + \operatorname{span}\{|m\rangle, m<n\}\}.  
    \end{array}
    \]
    Then $U_p(\mathcal{H})$, $A_p(\mathcal{H})$ and $N_p(\mathcal{H})$ are Banach Lie subgroups of $GL^p(\mathcal H)$ and the multiplication map
    \[
\textbf{m}:U_p(\mathcal{H})\times A_p(\mathcal{H})\times N_p(\mathcal{H})\rightarrow GL^p(\mathcal H), (u, a, g)\mapsto uag,
    \]
    is a diffeomorphism. In addition, both subgroups $A^p(\mathcal{H})$ and $N^p(\mathcal{H})$ are simply connected
and $A^p(\mathcal{H})N^p(\mathcal{H})= N^p(\mathcal{H})A^p(\mathcal{H})$.
\end{theorem}

\begin{remark}
    It was proved in \cite[Proposition~1.1]{beltita-iwasawa} that the multiplication map
    \[
\textbf{m}:U(\mathcal{H})\times A(\mathcal{H})\times N(\mathcal{H})\rightarrow GL(\mathcal H), (u, a, g)\mapsto uag,
    \]
    from the groups 
    \[
    \begin{array}{l}
    U(\mathcal{H}) = \{u \in GL(\mathcal H)~|~ u^* = u^{-1}\}\\
    A(\mathcal{H}) = \{ a \in GL(\mathcal H)~|~ a|n\rangle \in \mathbb{R}^{+*}|n\rangle\}\\
    N(\mathcal{H}) = \{ g \in GL(\mathcal H)~|~  g|n\rangle \in |n\rangle + \operatorname{span}\{|m\rangle, m<n\}.  
    \end{array}
    \]
    into $ GL(\mathcal H)$
    is bijective but not a diffeomorphism. This is related to the fact that the triangular truncation is unbounded on the space of bounded operators (see Example 4.1 in \cite{davidson}).
Let us also mention that the bijectivity of decompositions of Iwasawa type for invertible groups of hermitian algebras where obtained in \cite[Corollary~3.7]{beltita-neeb2010}. As far as we know, the existence of Iwasawa decomposition for the restricted group of invertible bounded operators on a polarized Hilbert space is an open question.

\end{remark}

\subsubsection{Invariant functions on $L^p(\mathcal H)$}

In order to apply Theorem~\ref{thm:aks_on_G}, we need to identify functions on $L^p(\mathcal H)$ which are invariant with respect to the adjoint action of the Banach Lie group $GL^p(\mathcal H)$. Note that for $1<p<+\infty$, every element $\mu\in L^p(\mathcal H)$ is compact, thus it can be represented in the form of a norm-convergent series
\[ \mu = \sum \lambda_i P_i \]
for some $\lambda_i\in\C$ and $\{P_i\}$ a sequence of mutually orthogonal projectors.
Thus a function which is invariant with respect to the action of $GL^p(\mathcal H)$ should only depend on eigenvalues $\lambda_i$ and their multiplicities $\dim P_i$. A family of such functions is  
\begin{equation}\label{eq:invariant_functions}
F_k(\mu) = \frac{1}{k+1}\Tr \mu^{k+1}, k\in\mathbb{N}, \mu \in L^p(\mathcal H).
\end{equation}

\subsubsection{Lax equations on the Manin triple $L^p(\mathcal{H}) = \mathfrak{u}_p(\mathcal{H})\oplus\b_p(\mathcal{H})$}

Combining Iwasawa decomposition of $GL^p(\mathcal H)$ given in Theorem~\ref{thm:Iwasawa} with the involutivity Theorem~\ref{thm:aks_on_G} for the Manin triple given in Proposition~\ref{triples}, we obtain the solutions of Lax equations on $L^p(\mathcal{H})$ for the family of  invariant functions defined by \eqref{eq:invariant_functions}. 

\begin{proposition}
    For $1<p\leq 2$, consider the Manin triple $L^p(\mathcal{H}) = \mathfrak{u}_p(\mathcal{H})\oplus\b_p(\mathcal{H})$ with $\Ad_{GL^p(\mathcal H)}$-invariant non-degenerate symmetric bilinear continuous map given by the imaginary part of the trace
    \begin{equation}
        \langle A, B\rangle = \im \Tr (AB), A, B \in L^p(\mathcal{H}).
    \end{equation}
    Let $R = p_{\mathfrak{u}_p} - p_{\mathfrak{b}_p}$ be the associated $R$-matrix with $p_{\mathfrak{u}_p}$ and $p_{\mathfrak{b}_p}$ the projections on $\mathfrak{u}_p(\mathcal{H})$ and $\b_p(\mathcal{H})$ with respect to the previous decomposition of $L^p(\mathcal{H})$. 
    Consider the family of spectral functions 
    \begin{equation}
    F_k(\mu) = \frac{1}{k+1}\Tr \mu^{k+1}, k\in\mathbb{N}, \mu \in L^p(\mathcal H).
    \end{equation}
    Then we have:
    \begin{enumerate}
        \item $\{F_{i}, F_j\}_R = 0, \forall i,j \in\mathbb{N}$
        \item the flow of the Hamiltonian vector field $X_{F_k}:= \{\cdot, F_k\}_R$ associated with $F_k$ with respect to the Poisson bracket $\pb_R$ satisfies the Lax equation
        \begin{equation}
            \frac{d\mu}{dt} = X_{F_k}(\mu) =   \left[\mu, p_{\mathfrak{u}_p}(\mu^{k})\right] = -\left[\mu, p_{\mathfrak{b}_p}(\mu^{k})\right].
        \end{equation}
        \item the integral curve of the Hamiltonian vector field $X_{F_k}$, starting at $\mu_0\in L^p(\mathcal{H})$, is given by
    \[
    \mu(t) = \Ad_{g_+(t)}\mu_0 = \Ad_{g_-(t)}\mu_0,
    \]
    where $g_+(t)\in U_p(\mathcal{H})$ and $g_-(t)\in B_p(\mathcal{H})$ are the smooth curves solving the factorization problem
    \begin{equation}
    \exp(-t  \mu_0^{k}) = g_+(t)^{-1} g_-(t), \textrm{  with initial conditions } g_\pm(0) = e.
    \end{equation}
    \end{enumerate}
\end{proposition}

\begin{proof}
Let $q$ be such that $\frac{1}{p}+\frac{1}{q} = 1$. Recall that for $1<p\leq 2$, $L^p(\mathcal{H})\subset L^q(\mathcal{H})$.
    Let us show that $R^*: (L^p(\mathcal{H}))^*\simeq L^q(\mathcal{H}) \rightarrow (L^p(\mathcal{H}))^*\simeq L^q(\mathcal{H})$ preserves $L^p(\mathcal{H})$. For $\mu \in L^q(\mathcal{H})$ and $A\in L^p(\mathcal{H})$ one has:
    \[
    \begin{array}{ll}
\im \Tr \left(\mu RA\right) & = \im \Tr \left(p_{\mathfrak{u}_q}(\mu) + p_{\mathfrak{b}_q}(\mu))(p_{\mathfrak{u}_p}(A)-p_{\mathfrak{b}_p}(A)\right) \\ & = \im \Tr\left(p_{\mathfrak{u}_q}(\mu) p_{\mathfrak{u}_p}(A) - p_{\mathfrak{u}_q}(\mu) p_{\mathfrak{b}_p}(A) + p_{\mathfrak{b}_q}(\mu) p_{\mathfrak{u}_p}(A) - p_{\mathfrak{b}_q}(\mu) p_{\mathfrak{b}_p}(A)\right)\\
& = \im \Tr\left(- p_{\mathfrak{u}_q}(\mu) p_{\mathfrak{b}_p}(A) + p_{\mathfrak{b}_q}(\mu) p_{\mathfrak{u}_p}(A)\right) \\ &
= - \im \Tr \left(p_{\mathfrak{u}_q}-p_{\mathfrak{b}_q}\right)(\mu) A,
\end{array}
    \]
    where we have used the isotropy of $\mathfrak{u}_q$ and $\mathfrak{b}_q$, and  $L^p(\mathcal{H})\subset L^q(\mathcal{H})$. Hence $R^* = p_{\mathfrak{u}_q}-p_{\mathfrak{b}_q}$. Therefore $R^*$ preserves $L^p(\mathcal{H})$. The rest follows from Theorem~\ref{thm:Iwasawa} and  Theorem~\ref{thm:aks_on_G}.
\end{proof}
\section{Toda lattice and upper and lower triangular operators in Schatten ideals}\label{Lptriangle}

\subsection{Decomposition into lower- and upper-triangular operators}\label{sec:dec_up_down}

Endow the separable Hilbert space $\mathcal{H}$ with an orthonormal basis $\{|n\rangle\}_{n=1}^{\infty}$. Consider the following Banach Lie subalgebras of $L^p(\mathcal{H})$
\[
\begin{array}{l}
L^p(\mathcal{H})_{0} = \{x \in L^p(\mathcal{H}), x(|n\rangle) \in \C|n\rangle\}\\   \quad \quad\quad \quad\quad \quad \textrm{(diagonal operators)}\\ \\ 
L^p(\mathcal{H})_{++} = \{x \in L^p(\mathcal{H}), x(|n\rangle) \in \textrm{span}\{|m\rangle, 1\leq m< n\}\}\\  \quad \quad \quad \quad\quad\quad  \textrm{(strictly upper triangular operators)}\\ \\
L^{p}(\mathcal{H})_{--} = \{\alpha\in L^{p}(\mathcal{H}), \alpha(|n\rangle) \in \textrm{span}\{|m\rangle, m>n\}\}\\  \quad \quad \quad \quad\quad  \quad\textrm{(strictly lower triangular operators)}\\ \\
L^p(\mathcal{H})_{-} = L^p(\mathcal{H})_{--} \oplus L^p(\mathcal{H})_{0}\\   \quad \quad\quad \quad\quad \quad \textrm{(lower triangular operators)}\\ \\ 
L^{p}(\mathcal{H})_{+} = L^p(\mathcal{H})_{++} \oplus L^p(\mathcal{H})_{0}\\  \quad \quad \quad \quad\quad\quad\textrm{(upper triangular operators)}.
\end{array}
\]
Since the projectors on the ``lower triangular part'' and ``upper triangular part'' are well-defined in $L^p(\mathcal{H})$ for $1<p<\infty$ and continuous (see e.g. \cite[Ch.~III, Theorem~6.2]{gohberg70}),
one has the following decompositions into sums of closed subalgebras
\begin{align}\label{Lp}
L^p(\mathcal{H}) =& L^p(\mathcal{H})_{-} \oplus L^p(\mathcal{H})_{++}\\
L^p(\mathcal{H}) =& L^p(\mathcal{H})_{+}\oplus L^p(\mathcal{H})_{--}.
\end{align} 
We will denote by $p_{L^p(\mathcal{H})_{-}}$, $p_{L^p(\mathcal{H})_{++}}$, $p_{L^{p}(\mathcal{H})_{+}}$ and $p_{L^{p}(\mathcal{H})_{--} }$ the projections with respect to these Banach decompositions. 

The trace pairing allows to identify $L^p(\mathcal{H})_{-} ^*$ with $L^p(\mathcal{H})^*/ \left(L^p(\mathcal{H})_{-}\right)^0 =  L^{q}(\mathcal{H})/  \left(L^p(\mathcal{H})_{-}\right)^0$, where 
\[
\left(L^p(\mathcal{H})_{-}\right)^0 = \{\alpha \in L^{q}(\mathcal{H}), \Tr\left(\alpha x\right) = 0,~~\forall x\in L^p(\mathcal{H})_{-}\}
= L^{q}(\mathcal{H})_{--}.
\]
Therefore we obtain,
\begin{equation}\label{Lp+}
L^p(\mathcal{H})_{-} ^*\simeq L^{q}(\mathcal{H})_{+}
\end{equation}
and analogously
\begin{equation}
L^p(\mathcal{H})_{--} ^*\simeq L^{q}(\mathcal{H})_{++}.
\end{equation} 
Thus $L^p(\mathcal{H})_{\pm}$ and $L^p(\mathcal{H})_{\pm\pm}$ are also reflexive Banach spaces and in consequence they are Banach Lie--Poisson spaces. The coadjoint action of an element $\alpha\in L^{q}(\mathcal{H})_{+}$ on $x\in L^p(\mathcal{H})_{-}$ can be expressed as:
\begin{equation}\label{coad_trian}
\ad_{\alpha}^*x = p_{L^p(\mathcal{H})_{-}}\left([x, \alpha]\right).
\end{equation}

\subsection{Lax equations associated with the decomposition  \texorpdfstring{$L^p(\mathcal{H}) = L^p(\mathcal{H})_{-}\oplus L^p(\mathcal{H})_{++}$}{Lp(H) = Lp(H)- + Lp(H)++}}

We will focus now on the $R$-matrix related to the decomposition $L^p(\mathcal{H}) = L^p(\mathcal{H})_{-}\oplus L^p(\mathcal{H})_{++}$ and its (pre)dual $L^{q}(\mathcal{H}) = L^{q}(\mathcal{H})_{+} \oplus L^{q}(\mathcal{H})_{--}$. Put
\[
R = p_{L^p(\mathcal{H})_{-}} - p_{L^p(\mathcal{H})_{++}}.
\]
Since $L^p(\mathcal{H})_{-}$ and $L^p(\mathcal{H})_{++}$ are two closed subalgebras of $L^p(\mathcal{H})$, it follows from Proposition~\ref{example_R} that $R$ is an $R$-matrix. Since $L^p$ spaces are reflexive, we immediately obtain the following:

\begin{proposition}
The Banach space $L^{q}(\mathcal{H})$ is a Banach Lie--Poisson space both for the usual Lie bracket on $L^{p}(\mathcal{H})$ and for the $R$-bracket on $L^{p}(\mathcal{H})$.
\end{proposition}
\begin{proof}
    The part concerning Lie--Poisson structure for the usual Lie bracket (i.e. commutator) is straightforward from Definition~\ref{def:blp} using reflexivity of $L^{p}(\mathcal{H})$. The claim for $R$-bracket follows from the fact that both $L^{q}(\mathcal{H})_+$ and $L^{q}(\mathcal{H})_{--}$ are Banach Lie--Poisson space as well. Thus $L^{q}(\mathcal{H})$ with the Lie--Poisson structure related to the $R$-bracket is a direct sum of $L^{q}(\mathcal{H})_+$ and $L^{q}(\mathcal{H})_{--}$, where we multiply the Poisson bracket by $-1$ in the second component.
\end{proof}

\begin{remark}
The dual maps of $p_{L^p(\mathcal{H})_{-}}$ and $p_{L^p(\mathcal{H})_{++}}$ are
$p_{L^p(\mathcal{H})_{-}}^* = p_{L^{q}(\mathcal{H})_{+}}$ and  $p_{L^p(\mathcal{H})_{++}}^* = p_{L^{q}(\mathcal{H})_{--} }$. 
Hence the dual map of $R$ is 
\[
R^* = p_{L^{q}(\mathcal{H})_{+}} - p_{L^{q}(\mathcal{H})_{--} },
\]
and is a $R$-matrix on $L^{q}(\mathcal{H})$ since $L^{q}(\mathcal{H})_{+}\oplus L^{q}(\mathcal{H})_{--} = L^{q}(\mathcal{H})$.
\end{remark}

Applying Theorem~\ref{thm:aks_on_G} to $L^p(\mathcal{H}) = L^p(\mathcal{H})_{-}\oplus L^p(\mathcal{H})_{++}$, we get the following Proposition.
\begin{proposition}
    For $1<p<+\infty$, consider the decomposition  $L^p(\mathcal{H}) = L^p(\mathcal{H})_{-}\oplus L^p(\mathcal{H})_{++}$ with $\Ad_{GL^p(\mathcal H)}$-invariant non-degenerate symmetric bilinear continuous map given by the trace
    \begin{equation}
        \langle A, B\rangle = \Tr (AB), A, B \in L^p(\mathcal{H}).
    \end{equation}
    Let $R = p_{L^p(\mathcal{H})_{++}} - p_{L^p(\mathcal{H})_{-}}$ be the associated $R$-matrix with $p_{L^p(\mathcal{H})_{++}}$ and $p_{L^p(\mathcal{H})_{-}}$ the projections on $L^p(\mathcal{H})_{++}$ and $L^p(\mathcal{H})_{-}$ with respect to the previous decomposition of $L^p(\mathcal{H})$. 
    Consider the family of spectral functions 
    \begin{equation}
    F_k(\mu) = \frac{1}{k+1}\Tr \mu^{k+1}, k\in\mathbb{N}, \mu \in L^p(\mathcal H).
    \end{equation}
    Then we have:
    \begin{enumerate}
        \item $\{F_{i}, F_j\}_R = 0, \forall i,j \in\mathbb{N}$
        \item the flow of the Hamiltonian vector field $X_{F_k}:= \{\cdot, F_k\}_R$ associated with $F_k$ with respect to the Poisson bracket $\pb_R$ satisfies the Lax equation
        \begin{equation}
            \frac{d\mu}{dt} = X_{F_k}(\mu) =   \left[\mu, p_{L^p(\mathcal{H})_{++}}(\mu^{k})\right] = -\left[\mu, p_{L^p(\mathcal{H})_{-}}(\mu^{k})\right].
        \end{equation}
        \item the integral curve of the Hamiltonian vector field $X_{F_k}$, starting at $\mu_0\in L^p(\mathcal{H})$, is given by
    \[
    \mu(t) = \Ad_{g_+(t)}\mu_0 = \Ad_{g_-(t)}\mu_0,
    \]
    where $g_+(t)\in \Id + L^p(\mathcal{H})_{++}$ and $g_-(t) \in \Id + L^p(\mathcal{H})_{-}$ are the smooth curves solving the factorization problem for $|t|$ small enough
    \begin{equation}
    \exp(-t  \mu_0^{k}) = g_+(t)^{-1} g_-(t), \textrm{  with initial conditions } g_\pm(0) = e,
    \end{equation}
    \end{enumerate}
\end{proposition}

\begin{remark}
    In this case, we do not know if we have a global decomposition of $GL^p(\mathcal{H})$ into the product of the groups of upper and lower triangular operators. In finite-dimension, this is known as the $LU$-factorization. However, as mentioned above, the solution of the factorization problem exists at least locally.
\end{remark}

\subsection{Semi-infinite Toda lattice}\label{sec:Toda}
Following e.g. \cite{magri90}, consider the Banach Lie algebra of upper-triangular operators with the decomposition 
\[
L^p(\mathcal H)_+=L^p(\mathcal H)_{++} \oplus L^p(\mathcal H)_0
\] 
(see section~\ref{sec:dec_up_down} for notations).
As a Banach Lie algebra, $L^p(\mathcal H)_+$ is generated by elements $\{|n\rangle\langle n|\,|\, n\in\N\} \cup \{|n\rangle\langle n+1| \,|\, n\in\N\}$. 
The predual of $L^p(\mathcal H)_+$ can be identified with $L^q(\mathcal H)_-$ using the trace, see \eqref{Lp+}.

Consider the following Nijenhuis operator $N$ on $L^p(\mathcal H)_+$ and its dual map $N^*$ on $L^q(\mathcal H)_-$:
\[
N=p_{L^p(\mathcal H)_{0}}
\textrm{    and    }
N^* = p_{L^q(\mathcal H)_{0}}.
\]
Due to reflexivity, $L^q(\mathcal H)_-$ is a Banach Lie--Poisson space both for the usual bracket on $L^p(\mathcal H)_+$ and $N$-bracket.

Denote by $x_{\a\b}$ an operator of the form
\[x_{\a\b}=\sum_{n\in \N} \a_n |n\rangle\langle n| + \b_n |n\rangle\langle n+1| \in L^p(\mathcal H)_+ \]
and by $\mu_{\q\p}$ an operator of the form
\[\mu_{\q\p}=\sum_{n\in \N} \q_n |n\rangle\langle n| + \p_n |n+1\rangle\langle n| \in L^q(\mathcal H)_- \]
for some sequences $\a,\b\in\ell^p$ and $\p,\q\in\ell^q$. 
By $M$ we will mean the Banach space spanned by all operators $x_{\a\b}$
\[ M=\{x_{\a\b} \,|\, \a,\b\in\ell^p\}\]
and by $M^*$ its dual space, i.e. the Banach space
\begin{equation}\label{M}
M^*=\{\mu_{\q\p} \,|\, \p,\q\in\ell^q\}.
\end{equation}

Let us identify a sequence $\a$ in $\ell^p$ with a diagonal operator in $L^p(\mathcal H)$ which we will denote with the same letter $\a$, and let $S$ denote a shift operator $S|n\rangle = |n+1\rangle$. Then we can use the notation from \cite{Oind} and write 
\[ x_{\a\b} = \a + \b S^*, \qquad \mu_{\q\p} = \q + S \p. \]
We give a couple of straightforward lemmas, which will simplify further computations. 

\begin{lemma}\label{lem:shift}
Let $\sigma$ be the shift operator in $\ell^p$ defined as
\[\sigma(\a)_n = \a_{n+1}.\]
Then we have 
\[\a S = S \sigma(\a).\] 
\end{lemma}

\begin{lemma}\label{lem:comm}
Let us introduce a forward and backward difference operators on $\ell^p$:
\begin{align*}\delta^+ &= \sigma - \Id,\\
\delta^- &= \Id - \sigma^*.\end{align*} 
Then one has the following commutator relations:
\begin{align*} [x_{\a0},\mu_{\q0}] &= 0,\\
[x_{\a0},\mu_{0\p}] &= \mu_{0,\delta^+(\a)\p},\\
[x_{0\b},\mu_{\q0}] &= x_{0,\b\delta^+(\q)},\\
[x_{0\b},\mu_{0\p}] &= -\mu_{\delta^-(\b\p),0}.
\end{align*}
\end{lemma}
\begin{proof}
Let us compute the first commutator using Lemma~\ref{lem:shift}:
\[ [x_{\a0},\mu_{0\p}] = [\a,S\p] = \a S\p - S\p\a = S\sigma(\a)\p-S\a\p = S\delta^+(\a)\p = \mu_{0,\delta^+(\a)\p}.\]
The other formulas follow analogously.
\end{proof}

\begin{proposition}
    Let $H$ be a smooth function on $L^q(\mathcal H)_-$ depending only on $\p$ and $\q$. Consider the Hamilton equations generated by $H$ on $L^q(\mathcal H)_-$ related to the $N$-bracket.
 Then the subspace $M^*$ defined by \eqref{M} is preserved by the flow of $H$ and the Hamilton equations restricted to $M^*$ assume the form
\begin{align*}
\dot \q_n &= \frac12\left(\p_{n-1}\frac{\partial H}{\partial \p_{n-1}} - \p_n \frac{\partial H}{\partial \p_n}\right),\\
\dot \p_n &= \frac12 \p_n \left(\frac{\partial H}{\partial \q_{n+1}} - \frac{\partial H}{\partial \q_n}\right).
\end{align*}
\end{proposition}
\begin{proof}
    The derivative of $H$ is the following form
\[
DH(\mu) = x_{\a\b},
\]
where $\a_n = \frac{\partial H}{\partial q_n}$ and $\b_n = \frac{\partial H}{\partial p_n}$ for $n\in\N$. Hamilton equations thus read
\[ 
\dot \mu = -(\ad^*_N)_{x_{\a\b}} \mu .
\]
Using formulas \eqref{coad_trian} and \eqref{coad_N} we can express them in the form
\[
\dot \mu = -\frac12\big(\ad^*_{Nx_{\a\b}} + [N^*,\ad^*_{x_{\a\b}}]\big) \mu = -\frac12 p_{L^q(\mathcal H)_-} \big( [Nx_{\a\b},\mu] + N^*[x_{\a\b},\mu] - [x_{\a\b},N^*\mu] \big).
\]
Finally using the explicit form of $N$ and $N^*$ and applying it to an element $\mu_{\q\p}\in M^*$ we obtain
\[ 
\dot \mu_{\q\p} = -\frac12 p_{L^q(\mathcal H)_-} \big( [x_{\a0},\mu_{\q\p}] + p_{L^q(\mathcal H)_{0}}[x_{\a\b},\mu_{\q\p}] -
[x_{\a\b},\mu_{\q0}]\big) = \]
\[
=-\frac12 p_{L^q(\mathcal H)_-} \big( [x_{\a0},\mu_{0\p}] + p_{L^q(\mathcal H)_{0}}([x_{\a0},\mu_{0\p}] + [x_{0\b},\mu_{\q0}] + [x_{0\b},\mu_{0\p}])-[x_{0\b},\mu_{\q0}]\big).
\]
Applying Lemma~\ref{lem:comm} the equations simplify to
\[ 
\dot \mu_{\q\p} = -\frac12\big( \mu_{0,\delta^+(\a)\p} 
- \mu_{\delta^-(\b\p),0}\big) = \frac12 \mu_{\delta^-(\b\p),-\delta^+(\a)\p}.
\]
Writing it in terms of the coordinate sequences $\q$ and $\p$ yields
  \begin{align*} 
  2\dot \q &= -\delta^-(\b\cdot\p), \\
  2\dot \p &= \p\cdot\delta^+\a.
  \end{align*}
Explicitly, in terms of the partial derivatives of the Hamiltonian, the equations look as follows:
  \begin{align*} 
  2\dot \q_n &= \b_n\p_n - \b_{n-1}\p_{n-1} =  \p_n \frac{\partial H}{\partial \p_n} - \p_{n-1}\frac{\partial H}{\partial \p_{n-1}},\\
  2\dot \p_n &= \p_n(\a_n - \a_{n+1}) = \p_n \left(\frac{\partial H}{\partial \q_n} - \frac{\partial H}{\partial \q_{n+1}}\right)
  \end{align*}
\end{proof}
for $n\in\N$.

\begin{corollary}
For the quadratic Hamiltonian
    \[
    H(\mu_{\q\p}) = -\sum_{n=0}^\infty \big( \q_n^2  + 2\p_n^2 \big)
    \]
one obtains the equations of the form
  \begin{align*}
  \dot \q_n &= 2(\p_n^2 - \p_{n-1}^2),\\
  \dot \p_n &= \p_n \left(\q_n-\q_{n+1}\right)
  \end{align*}
for $n\in\N$. These are the equations of the semi-infinite Toda lattice in Flaschka coordinates, see \cite[Section~2.3]{magri90} for a finite Toda lattice version.
\end{corollary}

\begin{remark}
For another approach to the Banach formulation of semi-infinite Toda lattice we refer to \cite[Section~5]{Oind}. Note though that the authors in that paper incorrectly assumed that the splitting \eqref{Lp} holds also for $p=1$. Another possibility is to use an infinite-dimensional version of \cite[Section~15.2.2]{morrison2013}.
\end{remark}

\section*{Acknowledgments}

This research was partially supported 2020 National Science Centre, Poland / Fonds zur Förderung der wissenschaftlichen Forschung, Austria grant ``Banach Poisson--Lie groups and integrable systems'' 2020/01/Y/ST1/00123, I-5015N. We are grateful to the anonymous referee for the careful review of our paper and comments that definitely contributed to improving the quality of the exposition.

% \bibliography{../literatura}

\end{document}